\documentclass[lettersize, journal]{IEEEtran}

\usepackage[utf8]{inputenc} 
\usepackage[T1]{fontenc}
\usepackage{ifthen}
\usepackage{array}
\usepackage{mathrsfs}
\usepackage{amsthm}
\usepackage{multirow}
\usepackage{multicol}
\usepackage[dvipsnames]{xcolor}
\usepackage{pgfplots}
\pgfplotsset{compat=1.18}
\usepackage{comment}
\usepackage{braket}
\usepackage{stmaryrd}
\usepackage{mathtools}
\usepackage{bm}
\usepackage{amsmath,amssymb,amsfonts}
\usepackage{graphicx}
\usepackage{textcomp}
\usepackage[textsize=tiny]{todonotes}
\usepackage{tikz}
\usetikzlibrary{patterns}
\usepackage{booktabs}
\usepackage[caption=false,font=footnotesize]{subfig}  

\newcommand{\F}{\mathbb{F}}      
\newcommand{\sym}[1]{#1}         
\newcommand{\6}{\mathbf } 

\makeatletter
\renewcommand*\env@matrix[1][*\c@MaxMatrixCols c]{%
  \hskip -\arraycolsep
  \let\@ifnextchar\new@ifnextchar
  \array{#1}}
\makeatother

\usepackage[ruled, vlined, linesnumbered]{algorithm2e}

\SetCommentSty{mycommfont}

\usepackage[noadjust]{cite} 

\theoremstyle{plain}
\newtheorem{The}{Theorem}
\newtheorem{Prop}{Proposition}
\newtheorem{Lem}{Lemma}
\newtheorem{Cor}{Corollary}

\newtheorem{Def}{Definition}

\theoremstyle{definition}
\newtheorem{Exa}{Example}

\theoremstyle{remark}

\newcommand{\rk}{\mathrm{rk}}
\newcommand{\vspan}{\mathrm{span}}
\newcommand{\rs}{\mathrm{row}}
\newcommand{\kr}{\mathrm{ker}}

\newcommand{\rsim}{\!\overset{\rs}{\sim}\!}

\newcommand{\CP}[1]{C_{\mathrm{#1}}^\perp}

\newcommand{\wt}[1]{\left\lVert{#1}\right\rVert}
\newcommand{\vstack}[2]{\left[\begin{array}{@{}c@{}}{#1}\\{#2}\end{array}\right]}

\usepackage[acronym,shortcuts]{glossaries}
\newacronym{CSS}{CSS}{Calderbank-Shor-Steane}
\newacronym{DPM}{DPM}{dyadic permutation matrix}
\newacronym{LER}{LER}{logical error rate}
\newacronym{MILP}{MILP}{mixed-integer linear program}
\newacronym{LDPC}{LDPC}{low-density parity-check}
\newacronym{QLDPC}{qLDPC}{quantum low-density parity-check}
\newacronym{BCH}{BCH}{Bose-Chaudhuri-Hocquenghem}
\newacronym{QC}{QC}{quasi-cyclic}
\newacronym{QC-LDPC}{QC-LDPC}{quasi-cyclic low-density parity-check}
\newacronym{QD}{QD}{quasi-dyadic}
\newacronym{QM}{QM}{quantum Margulis}
\newacronym{QD-LDPC}{QD-LDPC}{quasi-dyadic low-density parity-check}
\newacronym{BP}{BP}{belief propagation}
\newacronym{BP4}{BP$4$}{quaternary belief propagation}
\newacronym{CN}{CN}{check-node}
\newacronym{MS}{MS}{Min-Sum}
\newacronym{BP+OSD}{BP2+OSD}{belief propagation plus ordered-statistics decoding}
\newacronym{QECC}{QECC}{quantum error-correcting code}
\newacronym{QSC}{QSC}{quantum stabilizer code}
\newacronym{QEC}{QEC}{quantum error correction}
\newacronym{RM}{RM}{Reed-Muller}
\newacronym{MSD}{MSD}{magic state distillation}
\newacronym{GB}{GB}{generalized bicycle}
\newacronym{BB}{BB}{bivariate bicycle}
\newacronym{LP}{LP}{Lifted Product}
\newacronym{SD}{SD}{self-dual}
\newacronym{GHP}{GHP}{generalized hypergraph product}
\newacronym{DC}{DC}{dual-containing}
\newacronym{IS}{IS}{intersecting subsets}
\newacronym{HP}{HP}{hypergraph product}
\newacronym{LHCB}{LHCB}{left-hand conveyor belt}
\newacronym{GM}{GM}{generator matrix}
\newacronym{SD-XYZ}{SD-XYZ}{self-dual XYZ stabilizer}
\newacronym{LC}{LC}{local Clifford}
\newacronym{GRM}{GRM}{generalized Reed-Muller}
\newacronym
  [longplural={parity-check matrices},
   shortplural={PCMs}]
  {PCM}{PCM}{parity-check matrix}

\pdfmapfile{+custom.map}

\usepackage{mathabx}
\usepackage{pifont}

\usepackage{booktabs}

\usepackage{array}

\usepackage[normalem]{ulem}

\begin{document}

\title{
Quantum XYZ Stabilizer Codes 

\thanks{The work of Alessio Baldelli was partially supported by Agenzia per la Cybersicurezza Nazionale (ACN) under the programme for promotion of XL cycle PhD research in cybersecurity (CUP I32B24001750005).
The work of Davide Orsucci and Francisco L\'azaro is part of the HESC lighthouse project within Munich Quantum Valley initiative and is supported by the Bavarian state government with funds from the Hightech Agenda Bavaria.
 
Alessio Baldelli and Massimo Battaglioni are with the Department of Information Engineering, Università Politecnica delle Marche, 60131 Ancona, Italy.
Davide Orsucci and Francisco Lázaro are with the Institute of Communications and Navigation, German Aerospace Center (DLR), 82234 We{\ss}ling, Germany.
Corresponding author: Alessio Baldelli. E-mail: a.baldelli@pm.univpm.it.
}
}

\author{Alessio Baldelli, Davide Orsucci, Francisco Lázaro, Massimo Battaglioni}


\maketitle

\begin{abstract}
    Stabilizer codes are often constructed within the Calderbank--Shor--Steane (CSS) framework, where two mutually orthogonal binary classical codes define $X$ and $Z$-type stabilizer generators. 
    While this structure is algebraically convenient,  additional non-CSS constraints may help suppress low-weight logical operators and improve decoding performance in the finite-length regime. 
    We thus introduce quantum XYZ stabilizer codes, whose parity-check matrix (PCM) is built from three pairwise orthogonal binary PCMs associated with $X$-, $Y$-, and $Z$-type stabilizer generators. 
    A nontrivial point is that an XYZ code instance is not automatically genuinely non-CSS: the same stabilizer group may admit a CSS generating set. 
    We characterize this collapse, obtaining  algebraic and rank conditions for deciding when the $Y$-type checks are redundant and when they define genuinely non-CSS stabilizer constraints. 
    We also derive upper and lower bounds on the quantum minimum distance, including bounds for mixed Pauli logical operators. 
    The novel framework includes a known non-CSS topological code, namely the XYZ$^2$ hexagonal code, and yields also sparse finite-length quantum low-density parity-check (qLDPC) constructions from intersecting-subset and quasi-dyadic code families.
   Simulations under depolarizing code-capacity noise and quaternary belief propagation decoding show that the proposed XYZ qLDPC instances can outperform representative CSS qLDPC instances with similar finite-length parameters. 
\end{abstract}


\section{Introduction}
\label{sec:intro}

Over the last thirty years, many \ac{QEC} techniques have been developed, which use redundant physical qubits for the detection and correction of errors in such quantum systems. 
The framework almost universally employed to construct a \ac{QECC} is that of \emph{stabilizer codes}, which stands as the quantum equivalent of classical linear codes~\cite{gottesman1997stabilizercodesquantumerror}. 
The best-known approach for obtaining stabilizer codes is the \emph{\ac{CSS}} construction~\cite{Calderbank1996, Steane1996}.
In a \ac{CSS} code, the stabilizer generators can be divided into two sets: generators containing only $X$-type operators, up to identities, and generators containing only $Z$-type operators, up to identities. These two sets are derived from two classical codes that are mutually orthogonal.

Many widely used and extensively studied \acp{QECC} belong to \ac{CSS} families. 
Notable examples include Shor's 9-qubit code, the first proposed \ac{QECC}~\cite{shor1995scheme}, as well as several  \emph{topological} codes~\cite{bravyi1998quantum, dennis2002topological}, such as the toric code~\cite{kitaev1997quantum_1} and surface codes~\cite{fowler2012surface}.
Topological codes in two dimensions are fundamentally limited in the achievable minimum distance $d$ and coding rate $R=k/n$ (where $n$ is code length and $k$ the number of encoded qubits), since the quantity $d^2k/n$ is at most constant~\cite{bravyi2010}. 
Thus, motivated by the need to improve the rate--distance tradeoff, the efforts of the research community have shifted towards \emph{\ac{QLDPC}} codes, which are codes whose stabilizer generators are \emph{sparse}, or equivalently have low weight, but for which geometric locality is not imposed. 
This flexibility has motivated a line of work on \ac{QLDPC} code constructions with progressively improved tradeoffs between minimum distances and coding rates~\cite{Tillich_HGP, Pant_Kal_almost_linear, panteleev_good}.
However, asymptotic scaling does not necessarily predict finite-length performance: at short code lengths, specific instances from asymptotically weaker families may outperform instances from asymptotically stronger ones.

Most known \ac{QLDPC} constructions live within the \ac{CSS} framework. 
A question which arises naturally is what happens if this restriction is relaxed and Pauli-$Y$ operators are also employed in the definition of the stabilizer generators. 
This idea is at the basis, for example, of \emph{XYZ product codes}, first introduced in~\cite{maurice_thesis} as a generalization of \ac{HP} codes~\cite{Tillich_HGP}, and further investigated in~\cite{Leverrier_2022, liang2024_xyzproduct}, where the authors show that specific instances can achieve an almost-linear minimum distance and constant rate. 
A further example of a topological code involving $Y$-type stabilizers is the XYZ$^2$ hexagonal stabilizer code~\cite{Srivastava2022_xyzhexagonal}.

In this work, we investigate a structured extension of the \ac{CSS} framework that is related to, but structurally distinct from, the constructions discussed above. 
Unlike XYZ product codes, our framework does not impose a specific three-fold product structure; instead, we define the stabilizer \ac{PCM} directly from three binary classical \acp{PCM} associated with pure $X$-, $Y$-, and $Z$-type stabilizer generators and satisfying pairwise orthogonality conditions. 
In other words, the \ac{CSS} code is enriched with an additional set of pure $Y$-type checks. 
We call the resulting codes \emph{quantum XYZ stabilizer codes}. 
The proposed framework therefore occupies an intermediate position between \ac{CSS} and stabilizer codes: it allows genuinely non-\ac{CSS} stabilizer structures, while its description through three classical codes allows us to use classical coding-theoretic arguments for their construction and distance analysis. 
However, the presence of all three Pauli types alone does not guarantee that the XYZ code is genuinely non-\ac{CSS}.
Characterizing the conditions under which this occurs is one of the main objectives of this work.

\subsection{Main Motivation}
\label{subsec:motivation}

The \ac{CSS} framework offers a particularly convenient algebraic structure, but it also restricts the stabilizer configurations available for constructing a code with given parameters. Relaxing this restriction enlarges the code-design space and may provide additional freedom to exclude low-weight logical operators when designing codes with a  given block length and rate. This possibility is especially relevant at finite block lengths, where the detailed structure of individual code instances can be more important than the asymptotic scaling of the underlying family. 
Motivated by this observation, we investigate the following conjecture: quantum XYZ stabilizer code families can achieve larger minimum distance than comparable \ac{CSS} code families at the same block length and rate. In this work, we investigate this conjecture in the finite-length regime. To formalize the comparison, let $d_{\mathcal{S}}(n,k)$ and $d_{\mathrm{CSS}}(n,k)$ denote the largest minimum distance achievable by, respectively, stabilizer and \ac{CSS} codes with parameters $\llbracket n,k \rrbracket$. Since \ac{CSS} codes form a subclass of stabilizer codes, it trivially holds that~$d_{\mathrm{CSS}}(n,k) \leq d_{\mathcal{S}}(n,k)$.  Conversely, by the stabilizer-to-\ac{CSS} mapping of \cite[Theorem~1]{Kovalev_Pryadko}, any stabilizer code $\llbracket n,k,d \rrbracket$ can be mapped to a \ac{CSS} code $\llbracket 2n,2k,d' \rrbracket$ with $d \leq d' \leq 2d$ and stabilizer generator weights at most doubled\footnote{A related construction gives a \ac{CSS} code with parameters $\llbracket 4n, 2k, 2d \rrbracket$, again with at most doubled generator weights~\cite[Lemma~2]{Bravyi_MFCs_2010}.}. 
Hence, we get
\begin{align}
  d_{\mathrm{CSS}}(n,k) \leq  d_{\mathcal{S}}(n,k)  \leq  d_{\mathrm{CSS}}(2n,2k).
\end{align}
Therefore, defining the asymptotic relative distances as a function of the rate $R = k/n$, we get
\begin{align}
    d_{\ast}^{\mathrm{rel}}(R) &:= \limsup_{n\to\infty} \frac{1}{n} d_{\ast}(n,\lfloor Rn\rfloor),
\end{align}
with $\ast\in\{\mathcal{S}, \mathrm{CSS}\}$. As a consequence, we obtain
\begin{align}
\label{eq:rel_d}
    d_{\mathrm{CSS}}^{\mathrm{rel}}(R) \leq d_{\mathcal{S}}^{\mathrm{rel}}(R) \leq 2 d_{\mathrm{CSS}}^{\mathrm{rel}}(R).
\end{align}
Eq.~\eqref{eq:rel_d} shows that relaxing the \ac{CSS} restriction cannot yield more than a constant-factor improvement in the asymptotic relative distance. It does not, however, preclude substantial improvements for specific finite-length parameters. Although this comparison concerns general stabilizer codes, quantum XYZ stabilizer codes provide a structured non-\ac{CSS} setting in which such a finite-length advantage can be investigated using convenient analytical tools. Our objective is therefore to determine whether the additional design freedom offered by quantum XYZ stabilizer codes can produce better distance--rate tradeoffs in the short-block-length regime. Although a larger minimum distance does not necessarily imply better performance under every decoder and noise model, the numerical results presented later show that the proposed \ac{QLDPC} XYZ instances can achieve lower \ac{LER} than representative \ac{CSS} instances under low-complexity \ac{BP} decoding.

\subsection{Our Contribution} \label{subsec:Contribution}

As anticipated, our main contribution is the introduction of quantum XYZ stabilizer codes, a class of stabilizer codes whose \ac{PCM} is specified by three binary classical codes satisfying pairwise orthogonality conditions.

We first provide a structural characterization of this construction. 
We show how the row spans associated with the $X$-, $Y$-, and $Z$-type stabilizers decompose into reducible and irreducible components, thereby separating the part that is already generated by an underlying \ac{CSS} structure from the part that is genuinely XYZ. 
This leads to necessary and sufficient conditions for an XYZ code to be (or not to be) \ac{CSS}, and to rank criteria for detecting when the  $Y$-type checks do not introduce a genuinely new stabilizer structure. We also discuss the stronger notions of genuineness under uniform and local Pauli relabeling, clarifying that the  presence of $X$-, $Y$-, and $Z$-type generators is  not sufficient to obtain a genuinely quantum XYZ stabilizer code.

Then, we perform a theoretical analysis of the distance properties of XYZ codes. We derive upper bounds by restricting the logical-operator search to pure $X$-, $Y$-, and $Z$-type logical operators. We then derive lower bounds by associating to every XYZ code a collection of \ac{CSS} codes obtained by removing one irreducible Pauli component at a time. Since these \ac{CSS} codes contain the original XYZ codespace, their distances provide lower bounds on the XYZ minimum distance. Finally, we refine this analysis for mixed Pauli logical operators, obtaining a bound that explicitly uses the interaction between the three classical component codes. These results explain, at the level of logical operators, how additional $Y$-type stabilizer constraints can increase the minimum distance with respect to the CSS structures from which the construction is derived.

We also provide a constructive contribution.  
We first show that a known topological non-\ac{CSS} code, namely the XYZ$^2$ hexagonal code~\cite{Srivastava2022_xyzhexagonal}, fits naturally within the proposed XYZ framework. 
Then, we use our framework to design two new sparse finite-length \ac{QLDPC} XYZ code families, one based on \ac{IS} codes~\cite{Ostrev_IS} and one on \ac{QD} codes~\cite{baldelli2026ISIT}.

Finally, we evaluate the proposed \ac{QLDPC} XYZ instances over the depolarizing code-capacity channel using \ac{BP4} decoding. The simulations compare the new XYZ codes with representative \ac{CSS} \ac{QLDPC} codes of comparable block length, rate, and stabilizer generator weight. 
The results provide finite-length evidence that the additional $Y$-type constraints can translate the distance improvements suggested by our analysis into improved \ac{LER} performance under iterative decoding.

\subsection{Paper Organization} \label{subsec:paper_org}

The remainder of the paper is organized as follows. 
Section~\ref{sec:preli} introduces the notation and preliminary material used throughout the paper. Section~\ref{sec:quantum_XYZ_codes} defines quantum XYZ stabilizer codes from three pairwise orthogonal classical codes and studies their relation with \ac{CSS} codes. 
In particular, it introduces the reducible and irreducible components of the XYZ check spaces and derives algebraic and rank criteria for identifying genuinely non-\ac{CSS} instances. 
Section~\ref{sec:bounds_dmin} investigates the minimum distance of XYZ codes by deriving upper bounds from pure logical operators and lower bounds from \ac{CSS} codes imposed by the XYZ stabilizer structure, including bounds for mixed logical operators. 
Section~\ref{sec:XYZ_families} discusses families of stabilizer codes that fit the XYZ framework, including a known topological code instance and novel finite-length \ac{QLDPC} constructions based on \ac{IS} and \ac{QD} codes. 
Section~\ref{sec:num_res} presents numerical simulations over the code-capacity noise model using a \ac{BP4} decoder and compares the proposed XYZ \ac{QLDPC} instances with representative \ac{CSS} \ac{QLDPC} codes of comparable parameters. 
Finally, Section~\ref{sec:conclusions} concludes the paper and outlines possible directions for future work.

\section{Notation and Preliminaries}
\label{sec:preli}

For $a \in \mathbb{N}$, let $[a] := \{1,2,\dots,a\} \subset \mathbb{N}$. We use calligraphic uppercase letters to denote tuples, sets of subsets, and groups, e.g., $\mathcal{S}$. 
Given a group $\mathcal{G}$, the normalizer of a subgroup $\mathcal{S}\subseteq \mathcal{G}$ is denoted by $\mathcal{N}(\mathcal{S})$. 
Column vectors, e.g., $\6a$, $\6b$, and matrices, e.g., $\6A$, $\6B$, are denoted by bold lowercase and uppercase letters, respectively. 
Unless otherwise specified, the operations that involve binary vectors and matrices are performed over the binary field $\mathbb{F}_2$. 
We denote transposition by $(\cdot)^{\top}$.
The $m \times m$ identity matrix is denoted as $\6I_m$, and the subscript is omitted when it is clear from the context. 
Similarly, the  the all-zero  and all-one matrices (vectors) are denoted by $\60$ and  $\61$, respectively.
Given a matrix $\6A \in \mathbb{F}_2^{m\times n}$, we denote by $\rk(\6A)$ its rank over $\mathbb{F}_2$, while using $\kr(\6A):=\{\mathbf v\in\mathbb F_2^n \, | \, \mathbf A \mathbf v=\mathbf 0\}$, 
$\vspan(\6A) := \{\6A \6v \in \mathbb{F}_2^m \mid \6v \in \mathbb{F}_2^n \}$ and $\rs(\6A) = \vspan(\6A^\top)= \{(\6w^\top \6A)^\top \in \mathbb{F}_2^n \mid \6w \in \mathbb{F}_2^m \}$
to indicate the null space, (column) span, and row span of $\6A$, respectively; with this convention, a vector in the row span of $\6A$, $\6v \in \rs(\6A)$, is a column vector.
Moreover, we write $\6A \rsim \6B$ to indicate $\rs(\6A) = \rs(\6B)$. 
We omit round parentheses for functions of block matrices, e.g., $\vspan [\6A \, \6B] = \vspan([\6A \, \6B])$. The Hamming weight (or simply weight) of the vector $\6a$, i.e., the number of its non-zero entries, is denoted by $\wt{\6a}$. 

\subsection{Classical Linear Codes}
\label{subsec:classical_codes}
A \emph{classical binary linear code} $C$ with code length $n$, encoding $k$ bits, and having \textit{minimum (Hamming) distance} $\delta$ is a $k$-dimensional linear subspace of $\mathbb{F}_2^{n}$, denoted as $C[n, k, \delta]$.
The minimum distance $\delta$ is equal to the weight of a minimum-weight non-zero codeword $\6c \in C$.
The coding rate of $C$ is $R := k/n$. A linear code is defined as the null space of a \ac{PCM} $\6H\in\mathbb{F}_2^{m \times n}$, with $m \geq n-k = \rk(\6H)$ (where the \ac{PCM} may have linearly dependent rows), that is,
$C := \kr(\6H) = \left\{\6c \in\mathbb F_2^n\mid  \6H\6c = \60\right\}$,
where $\6{c} \in C$ is a codeword. Alternatively, we can define $C := \vspan(\6G)$, where $\6G \in \mathbb{F}_2^{n \times k}$ satisfying $\6H \6G = \60$ is a (full-rank) generator matrix of the code. We denote the \emph{dual code} of $C$ by $C^\perp[n, n - k, \delta^\perp]$, that is, 
$C^\perp := \left\{\6w \in\mathbb F_2^n \mid \6c^\top\6w = 0, \, \forall \, \6c \in C \right\}$.
With $C_1 + C_2$ we denote the linear subspace obtained by vector spaces addition, namely $C_1 + C_2 = \{ c_1 + c_2 \, | \, c_1 \in C_1, c_2 \in C_2 \}$ and with $C_1 \oplus C_2$ their \emph{direct sum}, corresponding to the addition of linearly independent subspaces.

\subsection{Quantum Information} \label{subsec:quantum_information}

We denote the single-qubit Pauli group as 
\begin{align}
    \mathcal{P}_1  :=  \{\phi\,\6P \mid \phi\in\Phi,\ \6P\in\{\6I,\6X,\6Y,\6Z\}\},
    \Phi := \{\pm 1, \pm i\},
\end{align} 
where $\6Y=i\6X\6Z$.\footnote{We will use the uppercase bold letters when we denote Pauli operators as matrices, e.g., $\6X$, $\6Y$, and uppercase italic letters when we want to specify the type of an operator, e.g., $\sym{X}$-type, $Y$-type.} We define a \emph{Pauli operator on $n$ qubits} as the $n$-fold tensor product of $n$ elements of $\mathcal{P}_1$, which can be expressed as $\6P = \phi \, \6P_1 \otimes \, \dots \, \otimes \6P_{n} \in (\mathbb{C}^{2 \times 2})^{\otimes n}$, where $\phi\in\Phi$ is the \emph{global phase}, and $\6P_i \in \{ \6I, \6X, \6Y, \6Z \}$. 
Then, the Pauli operators act on $n$ qubits and form the so-called \emph{n-qubit Pauli group} $\mathcal{P}_{n}$. 
In Pauli tensor products, the symbol $\otimes$ and identities will be omitted, and we will just indicate the qubit on which each Pauli is acting (e.g., $\6P = \6X_2\6Y_5$ acts with $\6{X}$ on qubit 2 and $\6{Y}$ on qubit 5). 
We also denote with $\text{wt}(\6P)$ the \emph{weight} of a Pauli operator $\6P \in \mathcal{P}_{n}$, which consists of the number of non-trivial, i.e., non-identity, elements in its associated tensor product decomposition.

Next, we consider the quotient Pauli group $\mathcal{P}_n/\Phi$, where Pauli operators differing only by a global phase are regarded as equivalent, on which we define the bijective map:
\begin{align} 
\label{eq:mapping}
    \notag
    \6{I} \mapsto [0 \, | \, 0], & \quad \6{X} \mapsto [1 \,  | \, 0], \quad \6{Z} \mapsto [0 \, | \, 1], \quad \6{Y} \mapsto [1 \, | \, 1], \\
    & \6P \mapsto \left[ \6u^\top \, | \, \6v^\top \right] = [u_1 \dots u_n \, | \, v_1 \dots v_n],
\end{align}
where $u_i, v_i \in \mathbb{F}_2$ and the index $i$ represents the $i$-th qubit.
This induces a group isomorphism between ${\mathcal P}_n/\Phi$, with the group operation inherited from multiplication in $\mathcal P_n$, and the additive group $\mathbb F_2^{2n}$ and a non-invertible (many-to-one) mapping from $\mathcal{P}_n$ to $\mathbb{F}_{2}^{2n}$. Note that 
\begin{align}
\label{eq:weight_identity}
    \text{wt}(\6P) = \wt{\6u \vee \6v} = \tfrac{1}{2}(\wt{\6u}+\wt{\6v}+\wt{\6w}),    
\end{align}
where $\vee$ denotes the bit-wise OR and $\6w=\6u+\6v$. The last equality follows from $\wt{\6u \vee \6v} = \wt{\6u} + \wt{\6v} - \wt{\6u \wedge \6v}$ and $\wt{\6u + \6v} = \wt{\6u} + \wt{\6v} - 2\wt{\6u \wedge \6v}$, where $\wedge$ denotes the bit-wise AND.

We call \emph{pure} $\sym{X}$-, $\sym{Y}$-, and $\sym{Z}$-type Pauli operators the $n$-qubit Pauli operators whose non-identity tensor factors are, respectively, only $\6X$, only $\6Y$, or only $\6Z$, independently of the global phase, and \emph{mixed} otherwise.
Two $n$-qubit Pauli operators $\6P_1, \6P_2 \in \mathcal{P}_n$ either commute ($\6P_1 \6P_2 = \6P_2 \6P_1$) or anti-commute ($\6P_1 \6P_2 = -\6P_2 \6P_1$). 
Using the mappings $\6P_1 \mapsto [\6u_1^\top | \6v_1^\top]$ and $\6P_2 \mapsto [\6u_2^\top | \6v_2^\top]$, we say that $\6P_1$ and $\6P_2$ commute iff $\6u_1 \cdot \6v_2 + \6v_1 \cdot \6u_2 = 0$.

\subsection{Stabilizer Codes} 
\label{subsec:stabilizer_codes}

Qubit stabilizer codes are the quantum analog of classical binary linear codes~\cite{gottesman1997stabilizercodesquantumerror}. 
\begin{Def} [Stabilizer Codes]
A \emph{stabilizer group}
$\mathcal{S} = \langle \6S_1, \dots, \6S_m \rangle$
is an abelian subgroup of $\mathcal{P}_n$ that does not contain $-\6I$, generated by $m \geq n-k$ commuting Pauli operators $\6S_1,\dots,\6S_m\in\mathcal P_n$, whose binary images under
\eqref{eq:mapping} have rank $n-k$. The associated \emph{stabilizer code} with $n$ physical qubits, $k$ logical qubits, and minimum quantum distance $d$,
denoted by $\mathcal{C}\llbracket n,k,d\rrbracket$, is defined as the $2^k$-dimensional subspace
\begin{align}
    \mathcal C
    :=
    \left\{
        \ket{\psi}\in(\mathbb C^2)^{\otimes n}
        \,\middle|\,
        \6S_i\ket{\psi}=\ket{\psi},\ \forall i\in[m]
    \right\}.    
\end{align}
\end{Def}

\begin{Def} [Degenerate Errors and Logical Operators]
For a stabilizer group $\mathcal{S} \subseteq \mathcal{P}_n$, the normalizer\footnote{For a stabilizer subgroup, the centralizer and the normalizer coincide~\cite{gottesman1997stabilizercodesquantumerror}.} $\mathcal{N}(\mathcal{S})$ contains all the \emph{degenerate errors} and \emph{logical errors} of the associated stabilizer code $\mathcal{C}$. Let
\begin{align}
    \Phi\mathcal{S} := \{\phi\6S \mid \phi\in\Phi, \,\6S\in\mathcal{S}\}.    
\end{align}
A Pauli operator $\6P \in \mathcal{S}\subseteq\mathcal{N}(\mathcal{S})$, by definition, acts trivially on the code space, i.e., $\6P\ket{\psi}=\ket{\psi}$ for all $\ket{\psi}\in\mathcal C$. More generally, every element of $\Phi\mathcal S$ acts on the code space as a global phase, namely, if $\6P=\phi\6S$ with $\phi\in\Phi$ and $\6S\in\mathcal S$, then $\6P\ket{\psi}=\phi\ket{\psi}$ for all $\ket{\psi}\in\mathcal C$. The elements of $\Phi\mathcal S$ are thus phase-equivalent to stabilizers and are known as \emph{degenerate} errors~\cite{Panteleev_degenerate} or \emph{harmless} errors~\cite{Babar2015}. Moreover, a Pauli operator $\6P \in \mathcal{N}(\mathcal{S}) \setminus \Phi\mathcal{S}$ commutes with all stabilizers but is not equivalent to a stabilizer up to a global phase. Hence, it preserves the code space but acts nontrivially on it, i.e., there exists $\ket{\psi}\in\mathcal C$ such that $\6P\ket{\psi} \neq \phi\ket{\psi},  \forall \phi\in\Phi$. These operators are called \emph{logical operators} or \emph{logical errors}~\cite{gottesman1997stabilizercodesquantumerror}.
\end{Def}

\begin{Def} [Minimum Distances of Stabilizer Codes] 
\label{def:min_dist_QSC}
The \emph{quantum minimum distance} $d$ of a stabilizer code is defined\footnote{If $k=0$, the distance is the minimum over an empty set, in which case we set $d=\infty$ by convention.} as follows
\begin{align} \label{eq:d_min_QSC}
    d := \min \{\textup{wt}(\6P) \, | \,
    \6P \in \mathcal{N}(\mathcal{S}) \setminus \Phi\mathcal{S}\}.
\end{align}
The \emph{classical minimum distance}~\cite{gottesman1997stabilizercodesquantumerror, Panteleev_degenerate}
$\delta$ of a stabilizer code is defined as the minimum weight of a non-scalar Pauli operator in the normalizer group:
\begin{align}
    \delta := \min \{\textup{wt}(\6P) \, | \,
    \6P \in \mathcal{N}(\mathcal{S}) \setminus \Phi\{\6I\} \}.
\end{align}
By construction, the quantum and classical minimum distances always satisfy $d \geq \delta$.
\end{Def}

By leveraging the isomorphism in~\eqref{eq:mapping} we can rewrite the stabilizer formalism in matrix form as follows. 

\begin{Prop} [Parity-Check Matrices for Stabilizer Codes]
\label{prop:PCM_stab}
Consider a code $\mathcal{C}\llbracket n, k, d \rrbracket$ with stabilizer group $\mathcal{S}$. Under~\eqref{eq:mapping}, the binary images of a possibly redundant set of $m\geq n-k$ stabilizer generators are collected as the rows of $\6H_{\mathcal S}$:
\begin{align} \label{eq:stab_matrix}
    \6H_{\mathcal{S}} :=
    \begin{bmatrix}[c|c]
        \6H_{\mathrm{X}} & \6H_{\mathrm{Z}}
    \end{bmatrix}
    \in \mathbb{F}_2^{m \times 2n},
\end{align}
where $\6H_{\mathrm{X}}, \6H_{\mathrm{Z}} \in \mathbb{F}_2^{m \times n}$. We call $\6H_{\mathcal{S}}$ a \ac{PCM} associated to $\mathcal{C}$.
The commutativity property of stabilizers, under the isomorphism \eqref{eq:mapping}, is equivalent to the \emph{symplectic product orthogonality condition}, that is:
\begin{align} \label{eq:symp}
    \6H_{\mathrm{X}} \6H_{\mathrm{Z}}^{\top} + \6H_{\mathrm{Z}} \6H_{\mathrm{X}}^{\top} = \60.
\end{align}

The dimension of the stabilizer code can be computed as 
\(k = n - \rk(\6H_{\mathcal{S}}),\) while its rate is $R = k/n$.
\end{Prop}

\begin{Prop} [Logical Operator Matrices for Stabilizer Codes]
The logical operators of a $\mathcal{C}\llbracket n, k, d \rrbracket$ stabilizer code admit a binary representation: the $2k$ independent logical generators are collected into the rows of the \emph{logical matrix}
\begin{align} \label{eq:logical_matrix}
    \6L_{\mathcal{S}} :=
    \begin{bmatrix}[c|c]
        \6L_{\mathrm{X}} & \6L_{\mathrm{Z}}
    \end{bmatrix}
    \in \mathbb{F}_2^{2k \times 2n},
\end{align}
where $\6L_{\mathrm{X}}, \6L_{\mathrm{Z}} \in \mathbb{F}_2^{2k \times n}$, and each row is the image under \eqref{eq:mapping} of a logical operator $\6L \in \mathcal{N}(\mathcal{S}) \setminus \Phi\mathcal{S}$.
By definition, these rows 
\begin{itemize}
    \item[$(i)$] commute with all stabilizer generators, i.e., they satisfy the symplectic orthogonality condition
    \begin{align} \label{eq:symp_logical}
        \6H_{\mathrm{X}} \6L_{\mathrm{Z}}^{\top} + \6H_{\mathrm{Z}} \6L_{\mathrm{X}}^{\top} = \60,
    \end{align}
    \item[$(ii)$] are not contained in $\rs(\6H_{\mathcal{S}})$.
\end{itemize}
The generators can be chosen as a symplectic (conjugate) basis $\{ \bar{\6X}_i, \bar{\6Z}_i \}_{i \in [k]}$, where $\bar{\6X}_i$ and $\bar{\6Z}_j$ anti-commute iff $i = j$, while all the remaining pairs commute. 
\end{Prop}

Note that, in the binary representation, the set $\Phi\mathcal S$ corresponds to $\rs(\6H_{\mathcal S})$, while the scalar subgroup $\Phi\{\6I\}$ corresponds to the zero vector.

\begin{Prop} [Minimum Distances of Stabilizer Codes] 
Using the notation of Proposition~\ref{prop:PCM_stab}, the quantum and the classical distances are given by:
\begin{subequations}
\label{eq:stab_distance}
\begin{align}
\label{eq:stab_distance_q}
    d = & \min\! \left\{ \wt{\6u \vee \6v} \, \middle| \, \6H_{\mathcal{S}}\! \vstack{\6v}{\6u} = \60, \vstack{\6u}{\6v} \notin \rs(\6H_{\mathcal{S}}) \right\}, \\
\label{eq:stab_distance_c}
    \delta = 
    & \min\! \left\{ \wt{\6u \vee \6v} \, \middle| \, \6H_{\mathcal{S}}\! \vstack{\6v}{\6u} = \60, \vstack{\6u}{\6v} \neq \60 \right\}.
\end{align}\end{subequations}
\end{Prop}

A special subclass of such stabilizer codes, in which the stabilizer generators consist exclusively of $\sym{X}$-type or $\sym{Z}$-type Pauli operators, is given by the so-called \ac{CSS} codes~\cite{Calderbank1996, Steane1996}.

\begin{Def} [\ac{CSS} Codes]
\label{def:CSS_code}
Consider two binary classical codes $C_{\mathrm{X}}[n, k_{\mathrm{X}}, \delta_{\mathrm{X}}]$ and $C_{\mathrm{Z}}[n, k_{\mathrm{Z}}, \delta_{\mathrm{Z}}]$ having associated \acp{PCM} $\6H_{\mathrm{X}} \in \mathbb{F}_2^{m_{\mathrm{X}} \times n}$ and $\6{H}_{\mathrm{Z}} \in \mathbb{F}_2^{m_{\mathrm{Z}} \times n}$, respectively, and such that $C_{\mathrm{Z}}^\perp \subseteq C_{\mathrm{X}}$, equivalently $C_{\mathrm{X}}^\perp \subseteq C_{\mathrm{Z}}$. 
This property can be expressed as 
\begin{align}
    \label{eq:ortho_CSS}
    \mathbf{H}_{\mathrm{Z}} \mathbf{H}_{\mathrm{X}}^{\top} = \60, \quad \text{equivalently} \quad \6H_{\mathrm{X}} \6{H}_{\mathrm{Z}}^{\top} = \60.
\end{align}

A \emph{\ac{CSS} code} can be represented by the following \ac{PCM}
\begin{align}
\label{eq:CSS_PCM}
    \6{H}_{\mathrm{CSS}} := 
        \begin{bmatrix} [c|c] 
            \6H_{\mathrm{X}} & \60 \\ 
            \60 & \6H_{\mathrm{Z}} 
        \end{bmatrix} \in \mathbb{F}_2^{(m_{\mathrm{X}} + m_{\mathrm{Z}}) \times 2n}
\end{align}
and we denote it as $\mathcal{C}(C_{\mathrm{X}}, C_{\mathrm{Z}})\llbracket n, k, d \rrbracket$, 
where $k=n-\rk(\6H_{\mathrm{X}})-\rk(\6H_{\mathrm{Z}})$ is the number of logical qubits, $R = k/n$ is the rate and $d$ is the minimum distance. 
Furthermore, the logical matrix of~\eqref{eq:logical_matrix} can be chosen to have the specific form:
\begin{align}
\label{eq:CSS_logical}
    \6{L}_{\mathrm{CSS}} := 
        \begin{bmatrix} [c|c] 
            \6L_{\mathrm{X}} & \60 \\ 
            \60 & \6L_{\mathrm{Z}} 
        \end{bmatrix} \in \mathbb{F}_2^{2k \times 2n}
\end{align}
where the $X$-type logical operators are associated to $\6L_{\mathrm{X}}$ and $Z$-type logical operators  are associated to $\6L_{\mathrm{Z}}$ (both consisting of $k$ rows) and satisfy:
\begin{align}
    \mathbf{H}_{\mathrm{Z}} \mathbf{L}_{\mathrm{X}}^{\top} = \60, \quad \6H_{\mathrm{X}} \6{L}_{\mathrm{Z}}^{\top} = \60,
    \quad \text{and} \quad \mathbf{L}_{\mathrm{Z}} \mathbf{L}_{\mathrm{X}}^{\top} = \6I.
\end{align}
\end{Def}

As a consequence, \ac{CSS} codes allow correcting $\sym{X}$-type and $\sym{Z}$-type errors separately. In particular, the classical code $C_{\mathrm{X}}$ with associated PCM $\6H_{\mathrm{X}}$ is used to correct $\sym{Z}$-type errors, and similarly for $C_{\mathrm{Z}}$ and $\6H_{\mathrm{Z}}$ with $X$-type errors.

\begin{Def} [Coset Distance]
\label{def:d_coset}
For subspaces $D \subseteq C \subseteq \mathbb{F}_2^n$ we define the \emph{coset minimum distance} $d(C\setminus D)$ and the minimum distance $d(C)$ as
\begin{subequations}
\begin{align}
    d(C\setminus D) & := \min\{ \wt{\6x} \,|\, \6x\in C,\ \6x\notin D \} \\
    d(C) & := \min\{ \wt{\6x} \,|\, \6x\in C,\ \6x\notin \{\60\} \} .
\end{align}\end{subequations}
By construction, the coset distance satisfies $d(C\setminus D) \geq d(C)$.
\end{Def}

\begin{Prop} [Minimum Distance of a \ac{CSS} Code]
\label{prop:CSS_d_min}
Let $\mathcal{C}(C_{\mathrm{X}}, C_{\mathrm{Z}})\llbracket n, k, d \rrbracket$ be a \ac{CSS} code constructed from classical codes $C_{\mathrm{X}}[n, k_{\mathrm{X}}, \delta_{\mathrm{X}}]$ and $C_{\mathrm{Z}}[n, k_{\mathrm{Z}}, \delta_{\mathrm{Z}}]$ as in Definition~\ref{def:CSS_code}. 
The lowest weights of $Z$-type and $X$-type logical operators are, respectively, given by the following \emph{coset distances}
\begin{subequations}\begin{align}
    d_{\mathrm{X} \setminus \mathrm{Z}} &:= d\big(\ker(\6H_{\mathrm{X}})\setminus \rs(\6H_{\mathrm{Z}})\big) = d(C_{\mathrm{X}}\setminus C_{\mathrm{Z}}^\perp),  \\
    d_{\mathrm{Z} \setminus \mathrm{X}} &:= d\big(\ker(\6H_{\mathrm{Z}})\setminus \rs(\6H_{\mathrm{X}})\big) = d(C_{\mathrm{Z}}\setminus C_{\mathrm{X}}^\perp),
\end{align}\end{subequations}
where $\ker(\6H_{\mathrm{X}}) = \rs(\6L_{\mathrm{Z}})+\rs(\6H_{\mathrm{Z}})$ and $\ker(\6H_{\mathrm{Z}}) = \rs(\6L_{\mathrm{X}})+\rs(\6H_{\mathrm{X}})$. Furthermore, we have
\begin{subequations}\begin{align}
    \delta_{\mathrm{X}} = d(C_{\mathrm{X}}) \leq d_{\mathrm{X} \setminus \mathrm{Z}}, \\
    \delta_{\mathrm{Z}} = d(C_{\mathrm{Z}}) \leq d_{\mathrm Z \setminus \mathrm X}.
\end{align}\end{subequations}
 Consequently, the quantum minimum distance $d$ and the classical distance $\delta$ satisfy:
\begin{align}
\label{eq:d_and_delta}
    d =  
    \min\! \left\{d_{\mathrm{Z} \setminus \mathrm{X}}, d_{\mathrm{X} \setminus \mathrm{Z}} \right\}
    \geq 
    \delta =
    \min\! \left\{\delta_{\mathrm{X}}, \delta_{\mathrm{Z}}\right\}.
\end{align}
\end{Prop}

See, e.g.,~\cite{Calderbank1996} for a proof.

\begin{Def} [Self-Dual \ac{CSS} Code]
A \ac{CSS} code $\mathcal{C}(C_{\mathrm{X}}, C_{\mathrm{Z}})$ is called self-dual (SD) if $C_{\mathrm{X}} = C_{\mathrm{Z}}$.
\end{Def}

\subsection{Error Model and Syndromes}
In this paper, we benchmark our codes using a \emph{code capacity noise model} which is based, specifically, on the \emph{depolarizing channel}. 
In this quantum channel, each qubit is affected by a $\6X$, $\6Z$, or $\6Y$ Pauli operator with probability $p/3$ each, and remains unaffected with probability $1 - p$, where $p$ is the \emph{depolarizing probability}. The syndromes associated to these errors are obtained as follows.

\begin{Def} [Syndrome Extraction]
Consider a stabilizer code with associated \ac{PCM} $\6H_\mathcal{S} = [\, \6H_{\mathrm{X}} \,\vert\, \6H_{\mathrm{Z}}\,] \in \mathbb{F}_2^{m\times2n}$ and a Pauli error with binary representation $\6e^\top=[\,\6e_{\mathrm{X}}^\top \,\vert\, \6e_{\mathrm{Z}}^\top\,] \in \mathbb{F}_2^{2n}$. 
The obtained \emph{syndrome vector} $\6s\in \mathbb{F}_2^{m}$ is
\begin{align}
    \6s = \6H_{\mathrm{X}}\6e_{\mathrm{Z}} + \6H_{\mathrm{Z}}\6e_{\mathrm{X}} .
\end{align}
\end{Def}

Note that with the block-diagonal structure of the \ac{PCM} of \ac{CSS} codes the syndrome outcome decouples as
\begin{align}
    \label{eq:CSS_syn}
    \6s_{\mathrm{X}} = \6H_{\mathrm{X}}\6e_{\mathrm{Z}}, \quad \6s_{\mathrm{Z}} = \6H_{\mathrm{Z}}\6e_{\mathrm{X}},
\end{align}
so that $\sym{X}$- and $\sym{Z}$-type errors can be corrected separately. This does not hold for the quantum XYZ stabilizer codes considered here, motivating the decoding strategy described in Section~\ref{subsec:deco_XYZ_codes}.

\subsection{Quantum LDPC Codes}
\label{subsec:LDPC_codes}

A quantum stabilizer code is referred to as \ac{QLDPC} when its \ac{PCM} is sparse. 
Its Tanner graph~\cite{Tanner1981} contains one variable node per qubit and one check node per stabilizer generator, with an edge whenever the generator acts nontrivially on the corresponding qubit. 
For non-\ac{CSS} codes, each edge is additionally labeled by its Pauli type. We denote the uniform stabilizer weight by $w_{\text{r}}$ (which is the row weight of the corresponding \ac{PCM}) and the girth of the Tanner graph associated with $\6H_\sigma$ by $g_\sigma$, with $\sigma \in \{ X, Y, Z \}$.

\section{Quantum XYZ Stabilizer Codes}
\label{sec:quantum_XYZ_codes}

In this Section, we introduce a novel framework for non-\ac{CSS} stabilizer codes, which includes an additional \ac{PCM} associated with stabilizers composed exclusively of $\sym{I}$- and $\sym{Y}$-type operators.  
\begin{Def} [Quantum XYZ Stabilizer Codes]
\label{def:XYZ_codes}
    Let us consider three classical linear codes
    $C_{\mathrm{X}}[n, k_{\mathrm{X}}, \delta_{\mathrm{X}}]$, $C_{\mathrm{Y}}[n, k_{\mathrm{Y}}, \delta_{\mathrm{Y}}]$, and $C_{\mathrm{Z}}[n, k_{\mathrm{Z}}, \delta_{\mathrm{Z}}]$ such that $C_{\mathrm{Y}}^\perp \subseteq C_{\mathrm{X}}$, $C_{\mathrm{Z}}^\perp \subseteq C_{\mathrm{Y}}$, and $C_{\mathrm{X}}^\perp \subseteq C_{\mathrm{Z}}$.
    These properties can be expressed in terms of \acp{PCM} representing such classical codes, using the following orthogonality conditions
    \begin{align} \label{eq:prod_XYZ}
        \6H_{\mathrm{X}} \6H_{\mathrm{Y}}^\top = \60, \quad \6H_{\mathrm{Y}} \6H_{\mathrm{Z}}^\top = \60, \quad \6H_{\mathrm{Z}} \6H_{\mathrm{X}}^\top = \60,
    \end{align}
    where $\6H_{\sigma} \in \mathbb{F}_2^{m_{\sigma} \times n}$, with $\sigma \in \{ X, Y, Z \}$.
    Under these conditions, we obtain a stabilizer code, described by the following \ac{PCM}
    \begin{align} \label{eq:PCM_XYZ}
        \6H_{\mathrm{XYZ}} := 
        \begin{bmatrix} [c|c]
            \6H_{\mathrm{X}} & \60 \\
            \60   & \6H_{\mathrm{Z}} \\
            \6H_{\mathrm{Y}} & \6H_{\mathrm{Y}}
        \end{bmatrix}
        \in \mathbb{F}_2^{(m_{\mathrm{X}} + m_{\mathrm{Y}} + m_{\mathrm{Z}}) \times 2n}.
    \end{align}
    We call it \emph{quantum XYZ stabilizer code} or simply \emph{XYZ code}, 
    and we denote it by $\mathcal{C}_{\mathrm{XYZ}}(C_{\mathrm{X}}, C_{\mathrm{Y}}, C_{\mathrm{Z}})\llbracket n, k, d \rrbracket$, or with the shorthand notations $\mathcal{C}_{\mathrm{XYZ}}(C_{\mathrm{X}}, C_{\mathrm{Y}}, C_{\mathrm{Z}})$ or $\mathcal{C}_{\mathrm{XYZ}}\llbracket n, k, d \rrbracket$, where the dimension is $k=n-\rk{(\6H_{\mathrm{XYZ}}})$.
\end{Def}

\subsection{Logical Operators and Syndromes of XYZ Codes} \label{subsec:deco_XYZ_codes}

Observe that \ac{CSS} codes can be chosen to have \ac{PCM} consisting only of $\sym{X}$- and $\sym{Z}$-type operators as in~\eqref{eq:CSS_PCM} and \emph{logical} generators consisting only of $\sym{X}$- and $\sym{Z}$-type operators as in~\eqref{eq:CSS_logical}. 
For quantum XYZ \emph{stabilizer} codes, the \ac{PCM} is extended by adding $\sym{Y}$-type stabilizers to~\eqref{eq:CSS_PCM}.\footnote{In principle, one could also consider quantum XYZ \emph{generator} codes, obtained by adding $Y$-type logical operators in~\eqref{eq:CSS_logical}, but these will not be considered in the present work.}
For XYZ codes, unlike in the \ac{CSS} case, representatives of the logical Pauli generators need not be $\sym{X}$-type or $\sym{Z}$-type up to identities; in general, they may mix $\sym{X}$, $\sym{Y}$, and $\sym{Z}$ components, as in~\eqref{eq:logical_matrix}.
Furthermore, for XYZ codes
the syndrome vector $\6s$ is given by the horizontal stack of
\begin{align}
    \6s_{\mathrm{X}} = \6H_{\mathrm{X}} \6e_{\mathrm{Z}}, \
    \6s_{\mathrm{Z}} = \6H_{\mathrm{Z}} \6e_{\mathrm{X}}, \
    \text{and} \
    \6s_{\mathrm{Y}} = \6H_{\mathrm{Y}} (\6e_{\mathrm{X}} + \6e_{\mathrm{Z}}).
\end{align}
Unlike in \ac{CSS} codes, $\sym{X}$- and $\sym{Z}$-type errors are mixed in $\6s_{\mathrm{Y}}$ and thus they must be decoded jointly.

\subsection{Matrix Decomposition of XYZ Codes} 
\label{subsec:deco_XYZ}

Notice that a $\sym{Y}$-type operator can be expressed as the sum of $X$- and $Z$-type operators. 
Therefore, the $\sym{Y}$-type checks that can be obtained from $X$- and $Z$-type stabilizers are redundant. 
These operators produce what we call the \emph{reducible}-XYZ component; the complementary part cannot be reduced to a CSS code and will be denoted as the \emph{irreducible}-XYZ component. 
More precisely, for an XYZ code $\mathcal{C}_{\mathrm{XYZ}}(C_{\mathrm{X}}, C_{\mathrm{Y}}, C_{\mathrm{Z}})$ having PCM as in~\eqref{eq:PCM_XYZ} we give some definitions. 
In particular, the following decomposition isolates the part of the $Y$-type checks that is already generated by the $X$- and $Z$-type checks. 

\begin{Def}[Reducible and Irreducible XYZ Components]
    Write $C_{\sigma}^{\perp} := \rs(\6H_{\sigma})$ for $\sigma\in\{\sym{X},\sym{Y},\sym{Z}\}$, then let $\6H_{\mathrm{X} \cap \mathrm{Y}}$, $\6H_{\mathrm{Y} \cap \mathrm{Z}}$, and $\6H_{\mathrm{X} \cap \mathrm{Z}}$ be any choice of full-rank binary matrices such that
    \begin{subequations}
    \begin{align}
        \rs(\6H_{\mathrm{X} \cap \mathrm{Y}}) 
        &= C_{\mathrm{X}}^{\perp} \cap C_{\mathrm{Y}}^{\perp}, 
        \\
        \rs(\6H_{\mathrm{Y} \cap \mathrm{Z}}) &= C_{\mathrm{Y}}^{\perp} \cap C_{\mathrm{Z}}^{\perp}, \\
        \rs(\6H_{\mathrm{X} \cap \mathrm{Z}}) 
        &= C_{\mathrm{X}}^{\perp} \cap C_{\mathrm{Z}}^{\perp}.   
    \end{align}    
    \end{subequations}
    Next, define $\6H_{\mathrm{X}}'$, $\6H_{\mathrm{Y}}'$, and $\6H_{\mathrm{Z}}'$ as any matrices such that
    \begin{subequations}
    \begin{align}
        \rs(\6H_{\mathrm{X}})
        & = 
        \rs(\6H_{\mathrm{X}}') 
        \oplus
        \rs\!
        \begin{bmatrix} [c]
        \6H_{\mathrm{X} \cap \mathrm{Y}}\\
        \6H_{\mathrm{X} \cap \mathrm{Z}}
        \end{bmatrix}, 
        \\
        \rs(\6H_{\mathrm{Y}})
        & =
        \rs(\6H_{\mathrm{Y}}')
        \oplus
        \rs\!
        \begin{bmatrix} [c]
        \6H_{\mathrm{X} \cap \mathrm{Y}}\\
        \6H_{\mathrm{Y} \cap \mathrm{Z}}
        \end{bmatrix}, \\
        \rs(\6H_{\mathrm{Z}})
        & = 
        \rs(\6H_{\mathrm{Z}}') 
        \oplus
        \rs\!
        \begin{bmatrix} [c]
        \6H_{\mathrm{X} \cap \mathrm{Z}}\\
        \6H_{\mathrm{Y} \cap \mathrm{Z}}
        \end{bmatrix}, 
    \end{align}    
    \end{subequations}
    which are direct sums of linearly independent subspaces. We then say that any matrix $\6H_{\mathrm{XYZ}}'$ such that
    \begin{align}
    \label{eq:irred}
        \rs 
        (\6H_{\mathrm{XYZ}}') 
        = \rs\!
        \begin{bmatrix} [c|c]
            \6H_{\mathrm{X}}' & \60 \\
            \60 & \6H_{\mathrm{Z}}' \\
            \6H_{\mathrm{Y}}' & \6H_{\mathrm{Y}}'
        \end{bmatrix} 
    \end{align}
    is a generator of the \emph{irreducible} XYZ component of the code. 
    Furthermore, we say that any matrix $\6H_{\cap}$ such that
    \begin{align}
    \label{eq:red}
        \rs 
        (\6H_\cap)
        = \rs\!
        \begin{bmatrix}
            \6H_{\mathrm{X} \cap \mathrm{Y}} \\
            \6H_{\mathrm{Y} \cap \mathrm{Z}} \\
            \6H_{\mathrm{X} \cap \mathrm{Z}}
        \end{bmatrix} 
        =:
        C_\cap^{\perp} 
    \end{align}
    is a generator of the \emph{reducible} XYZ component of the code.
\end{Def}

By exploiting the above definition, the \ac{PCM} of an XYZ code can be rewritten as: 
\begin{align}
\label{eq:red_irred}
    \6H_{\mathrm{XYZ}} 
    \rsim
    \begin{bmatrix} [c|c]
        \6H_{\mathrm{X}}' & \60 \\
        \6H_{\mathrm{X} \cap \mathrm{Y}} & \60 \\
        \6H_{\mathrm{X} \cap \mathrm{Z}} & \60 \\
        \60 & \6H_{\mathrm{Z}}' \\
        \60 & \6H_{\mathrm{Y} \cap \mathrm{Z}} \\
        \60 & \6H_{\mathrm{X} \cap \mathrm{Z}} \\
        \6H_{\mathrm{Y}}' & \6H_{\mathrm{Y}}' \\
        \6H_{\mathrm{Y} \cap \mathrm{Z}} & \6H_{\mathrm{Y} \cap \mathrm{Z}} \\
        \6H_{\mathrm{X} \cap \mathrm{Y}} & \6H_{\mathrm{X} \cap \mathrm{Y}}
    \end{bmatrix}
    \rsim
    \left[ \begin{array} {c|c}
        \6H_{\mathrm{X}}' & \60 \\
        \60 & \6H_{\mathrm{Z}}' \\
        \6H_{\mathrm{Y}}' & \6H_{\mathrm{Y}}' \\
        \hline
        \6H_\cap & \60 \\
        \60 & \6H_\cap
    \end{array} \right]
    .
\end{align}
Notice that the two lowest blocks contain the same matrix $\6H_\cap$ in both the $\sym{X}$ and the $\sym{Z}$ components, so these can be interpreted simultaneously as $\sym{X}$-, $\sym{Y}$-, and $\sym{Z}$-type checks. 
In contrast, the row spans of the matrices $\6H_{\sigma}'$ are mutually linearly independent and each generates $\sigma$-type checks separately, meaning that, e.g., $Y$-type checks cannot be obtained by combining together $\sym{X}$-type checks generated by $\6H_{\mathrm{X}}'$ and  $\sym{Z}$-type checks generated by $\6H_{\mathrm{Z}}'$. This matrix decomposition will be useful to identify when an XYZ code is \emph{genuine} and to study the minimum distance bounds of XYZ codes, respectively in Section~\ref{subsec:gen_XYZ_codes} and in Section~\ref{sec:bounds_dmin}.

\subsection{Genuine XYZ Codes} 
\label{subsec:gen_XYZ_codes}

In this section, we explore the connection between \ac{CSS} and XYZ codes. 
Although the \ac{PCM} of an XYZ code is built from three Pauli types, it may be a \ac{CSS} code in disguise, since it is possible that its stabilizer group is generated using only two Pauli types. 
This could be done either directly or after the application of a \ac{LC} unitary, i.e., a Clifford gate that acts on each single-qubit as a relabeling of the Pauli operators. These notions are formalized as follows.

\begin{Def}[Genuine XYZ Codes] 
\label{def:genuine_XYZ}
    The code $\mathcal{C}_{\mathrm{XYZ}}$ is called \emph{genuine} if it is not a \ac{CSS} code, and \emph{non-genuine} otherwise. Equivalently, the code is non-genuine if it admits a set of stabilizer generators that consists exclusively of pure $\sym{X}$- and $\sym{Z}$-type operators.
\end{Def}

\begin{Def}[Genuine XYZ Codes under Uniform Pauli Relabeling]
\label{def:genuine_XYZ_uniform}
    The code $\mathcal{C}_{\mathrm{XYZ}}$ is called \emph{genuine under uniform Pauli relabeling} if, for every single-qubit Clifford unitary $U$ (acting as a permutation of the $\{\sym{X}, \sym{Y}, \sym{Z}\}$ labels), $U^{\otimes n}\mathcal{C}_{\mathrm{XYZ}}$ is not a \ac{CSS} code, and is called \emph{non-genuine} otherwise. 
    Equivalently, a non-genuine XYZ code under uniform Pauli relabeling admits a set of stabilizer generators that consists exclusively of pure operators of $\sym{X}$- and $\sym{Y}$-type, or $\sym{Y}$- and $\sym{Z}$-type, or $\sym{X}$- and $\sym{Z}$-type.
\end{Def}

\begin{Def}[Genuine XYZ Codes under Local Pauli Relabeling]
\label{def:genuine_XYZ_local}
    The code $\mathcal{C}_{\mathrm{XYZ}}$ is called \emph{genuine under local Pauli relabeling} if, for all single-qubit Clifford unitaries $U_1, \dots, U_n$, the code $(U_1 \otimes \dots \otimes U_n)\mathcal{C}_{\mathrm{XYZ}}$ is not a \ac{CSS} code, and is called \emph{non-genuine} otherwise. 
    Equivalently, a non-genuine XYZ code under local Pauli relabeling is \ac{LC}-equivalent to a \ac{CSS} code, i.e., it can be transformed into it by a \ac{LC} unitary.
\end{Def}

With the next theorem, we show a necessary and sufficient condition for an XYZ code to be a \ac{CSS} code.

\begin{The}
\label{theo:gen_XYZ_codes}
    Let $\mathcal{C}_{\mathrm{XYZ}}(C_{\mathrm{X}}, C_{\mathrm{Y}}, C_{\mathrm{Z}})$ be an XYZ code. The code is \ac{CSS} (and thus non-genuine) iff
    \begin{align}
    \label{eq:CSS_eq_0}
        C_{\mathrm{Y}}^{\perp}
        =
        \left(C_{\mathrm{X}}^{\perp}\cap C_{\mathrm{Y}}^{\perp}\right)
        +
        \left(C_{\mathrm{Y}}^{\perp}\cap C_{\mathrm{Z}}^{\perp}\right).
    \end{align}
    In this case $\mathcal{C}_{\mathrm{XYZ}}(C_{\mathrm{X}}, C_{\mathrm{Y}}, C_{\mathrm{Z}})$ is a \ac{CSS} code with $\sym{X}$-type checks spanning $C_{\mathrm{X}}^\perp+(C_{\mathrm{Y}}^{\perp}\cap C_{\mathrm{Z}}^{\perp})$ and $\sym{Z}$-type checks spanning $C_{\mathrm{Z}}^\perp + (C_{\mathrm{X}}^{\perp}\cap C_{\mathrm{Y}}^{\perp})$. 
    Condition~\eqref{eq:CSS_eq_0} is equivalent to
    \begin{align}
    \label{eq:CSS_rank_condition}
        \rk\left(\6H_{\mathrm{XYZ}}\right)
        =
        \rk\!\begin{bmatrix}\6H_{\mathrm{X}}\\\6H_{\mathrm{Y}}\end{bmatrix}
        + \rk\!\begin{bmatrix}\6H_{\mathrm{Y}}\\\6H_{\mathrm{Z}}\end{bmatrix},
    \end{align}
    which can be numerically verified efficiently.
\end{The}

The rank condition~\eqref{eq:CSS_rank_condition} can also be obtained directly from~\cite[Lemma7.4]{cross2025small}.

\begin{IEEEproof}
By construction we have
\begin{align}
    C_{\mathrm{Y}}^{\perp} 
    & = \rs(\6H_{\mathrm{Y}}') \oplus [\rs(\6H_{\mathrm{X}\cap\mathrm{Y}}) + \rs(\6H_{\mathrm{Y}\cap\mathrm{Z}})] \nonumber\\
    & = \rs(\6H_{\mathrm{Y}}') \oplus [(C_{\mathrm{X}}^{\perp} \cap C_{\mathrm{Y}}^{\perp}) + (C_{\mathrm{Y}}^{\perp}\cap C_{\mathrm{Z}}^{\perp})].
\end{align}
Hence, the condition $C_{\mathrm{Y}}^{\perp} =(C_{\mathrm{X}}^{\perp} \cap C_{\mathrm{Y}}^{\perp}) + (C_{\mathrm{Y}}^{\perp} \cap C_{\mathrm{Z}}^{\perp})$ is equivalent to $\6H_{\mathrm{Y}}'=\60$.

If $\6H_{\mathrm{Y}}'=\60$, Eq.~\eqref{eq:red_irred} directly shows that the code is \ac{CSS}. 

We now show, conversely, that if the code is \ac{CSS} then $\6H_{\mathrm{Y}}'=\60$. We consider the spaces
\begin{align}
    N_{\mathrm{X}} 
    & := C_{\mathrm{X}}^{\perp}+(C_{\mathrm{Y}}^{\perp}\cap C_{\mathrm{Z}}^{\perp}) =
    \rs(\6H_{\mathrm{X}}')+\rs(\6H_{\cap}), \\
    N_{\mathrm{Z}}
    & := C_{\mathrm{Z}}^{\perp}+(C_{\mathrm{X}}^{\perp}\cap C_{\mathrm{Y}}^{\perp})
    = \rs(\6H_{\mathrm{Z}}')+\rs(\6H_{\cap}).
\end{align}
We then apply the \emph{modular law} $(D + E) \cap F = D + (E \cap F)$, which holds for vector spaces $D \subseteq F$. Taking $D = C_{\mathrm{Y}}^{\perp} \cap C_{\mathrm{Z}}^{\perp}$, $E = C_{\mathrm{X}}^{\perp}$, and $F = C_{\mathrm{Y}}^{\perp}$ we get
\begin{align} 
    N_{\mathrm{X}}\cap\rs(\6H_{\mathrm{Y}}')
    & \subseteq 
    N_{\mathrm{X}} \cap C_{\mathrm{Y}}^{\perp} \nonumber\\
    & = 
    (C_{\mathrm{Y}}^{\perp}\cap C_{\mathrm{Z}}^{\perp}) + (C_{\mathrm{X}}^{\perp}\cap C_{\mathrm{Y}}^{\perp}).
\end{align}
Since $(C_{\mathrm{Y}}^{\perp}\cap C_{\mathrm{Z}}^{\perp}) + (C_{\mathrm{X}}^{\perp}\cap C_{\mathrm{Y}}^{\perp})$ intersects $\rs(\6H_{\mathrm{Y}}')$ only in $\60$, we get $N_{\mathrm{X}} \cap \rs(\6H_{\mathrm{Y}}') = \{\60\}$. Similarly we obtain $N_{\mathrm{Z}} \cap \rs(\6H_{\mathrm{Y}}') = \{\60\}$.
The row span $W := \rs(\6H_{\mathrm{XYZ}})$ can then be written as
\begin{align}
    W & = 
    \left\{ 
    \vstack{\6x+\6y}{\6z+\6y}
    \,\middle|\, 
    \6x \in N_{\mathrm{X}}, \6z\in N_{\mathrm{Z}}, \6y \in \rs(\6H_{\mathrm{Y}}') \right\}, \\
    W_{\mathrm{X}}
    & :=
    W\cap\left(\mathbb{F}_2^{n}\times\{\60\}\right),
    \
    W_{\mathrm{Z}}
    := 
    W\cap\left(\{\60\}\times\mathbb{F}_2^{n}\right).
\end{align}
A vector belongs to $W_{\mathrm{X}}$ (i.e., it is of $\sym{X}$-type) iff $\6y+\6z=\60$, that is, $\6y=\6z \in N_{\mathrm{Z}}\cap\rs(\6H_{\mathrm{Y}}')=\{\60\}$, which implies $\6y=\6z=\60$. Similarly, a vector in $W$ is of $\sym{Z}$-type iff $\6x=\6y =\60$. Moreover, the $\sym{X}$, $\sym{Y}$, and $\sym{Z}$ components are independent, since
\begin{align}
\label{eq:linear_indip}
    \vstack{\6x}{\60} +
    \vstack{\60}{\6z} +
    \vstack{\6y}{\6y} =
    \vstack{\60}{\60}
\end{align}
implies $\6y = \6x\in N_{\mathrm{X}} \cap \rs(\6H_{\mathrm{Y}}') = \{\60\}$, which forces $\6x=\6y=\6z=\60$. Consequently we have:
\begin{align}
\label{eq:rk_split}
    \rk(\6H_{\mathrm{XYZ}}) = \dim(W_{\mathrm{X}}) + \dim(W_{\mathrm{Z}}) + \rk(\6H_{\mathrm{Y}}').
\end{align}
Since $W_{\mathrm{X}} \oplus W_{\mathrm{Z}}\subseteq W$ always holds, the code is \ac{CSS} iff $W = W_{\mathrm{X}}\oplus W_{\mathrm{Z}}$, i.e.\ iff $\rk(\6H_{\mathrm{XYZ}}) = \dim(W_{\mathrm{X}}) + \dim(W_{\mathrm{Z}})$,
which by~\eqref{eq:rk_split} is equivalent to $\6H_{\mathrm{Y}}'=\60$. This proves the converse implication.

To prove the rank condition~\eqref{eq:CSS_rank_condition} we consider the projection $\pi_{\mathrm{X}}: W \to \mathbb{F}_2^{n}$ onto the first $n$ components of the vector space; it has image $C_{\mathrm{X}}^{\perp}+C_{\mathrm{Y}}^{\perp}$ and kernel $W_{\mathrm{Z}}$, thus
\begin{align}
    \dim(N_{\mathrm{Z}})
    & =
    \dim(W_{\mathrm{Z}}) = 
    \rk(\6H_{\mathrm{XYZ}}) - 
    \rk\!
    \begin{bmatrix}[c]
        \6H_{\mathrm{X}}\\
        \6H_{\mathrm{Y}}
    \end{bmatrix} 
\end{align}
and a similar expression holds using the projection $\pi_{\mathrm{Z}}$ on the last $n$ components. Substituting into~\eqref{eq:rk_split} yields
\begin{align}
    \rk(\6H_{\mathrm{Y}}')
    = 
    \rk\!
    \begin{bmatrix}[c]
        \6H_{\mathrm{X}}\\
        \6H_{\mathrm{Y}}
    \end{bmatrix}
    +
    \rk\!
    \begin{bmatrix}[c]
        \6H_{\mathrm{Y}}\\
        \6H_{\mathrm{Z}}
    \end{bmatrix}
    - 
    \rk(\6H_{\mathrm{XYZ}}),
\end{align}
and thus $\rk(\6H_{\mathrm{Y}}') = 0$ is equivalent to condition~\eqref{eq:CSS_rank_condition}.
\end{IEEEproof}

We now consider the notion of genuine and non-genuine XYZ codes under uniform Pauli relabeling.

\begin{Cor} [Genuine XYZ Codes under Uniform Pauli Relabeling]
\label{cor:genuine_uniform}
By the same arguments as in Theorem~\ref{theo:gen_XYZ_codes}, the code $\mathcal{C}_{\mathrm{XYZ}}(C_{\mathrm{X}}, C_{\mathrm{Y}}, C_{\mathrm{Z}})$ is \ac{CSS} up to uniform relabeling iff \emph{at least} one of the following conditions holds:
\begin{subequations}
\label{eq:CSS_eq_1}
\begin{align}
    \label{eq:CSS_eq_1a}
    C_{\mathrm{X}}^{\perp}
    & =
    \left(C_{\mathrm{X}}^{\perp}\cap C_{\mathrm{Y}}^{\perp}\right)
    +
    \left(C_{\mathrm{X}}^{\perp}\cap C_{\mathrm{Z}}^{\perp}\right)\\
    \label{eq:CSS_eq_1b}
    C_{\mathrm{Y}}^{\perp}
    & =
    \left(C_{\mathrm{X}}^{\perp}\cap C_{\mathrm{Y}}^{\perp}\right)
    +
    \left(C_{\mathrm{Y}}^{\perp}\cap C_{\mathrm{Z}}^{\perp}\right) \\
    \label{eq:CSS_eq_1c}
    C_{\mathrm{Z}}^{\perp}
    & =
    \left(C_{\mathrm{X}}^{\perp}\cap C_{\mathrm{Z}}^{\perp}\right)
    +
    \left(C_{\mathrm{Y}}^{\perp}\cap C_{\mathrm{Z}}^{\perp}\right).
\end{align}
\end{subequations}
These conditions can be efficiently evaluated using the rank condition in~\eqref{eq:CSS_rank_condition} and the ones obtained via label permutations of $\{\sym{X},\sym{Y},\sym{Z}\}$.
\end{Cor}

\begin{Exa} [Non-genuine XYZ Codes from SD \ac{CSS} Code]
\label{exa:genuine}
    Consider $\mathcal{C}_{\mathrm{XYZ}}(C_{\mathrm{X}}, C_{\mathrm{X}}, C_{\mathrm{Z}})$, obtained by setting   $C_{\mathrm{Y}} = C_{\mathrm{X}}$, so that the construction is based only on two classical codes $C_{\mathrm{X}}$ and $C_{\mathrm{Z}}$. Its \ac{PCM} is
    \begin{align}
        \6H_{\mathrm{XYZ}} =
        \begin{bmatrix} [c|c]
            \6H_{\mathrm{X}} & \60 \\
            \60 & \6H_{\mathrm{Z}} \\
            \6H_{\mathrm{X}} & \6H_{\mathrm{X}}
        \end{bmatrix}
        \in \mathbb{F}_2^{(2m_{\mathrm{X}} + m_{\mathrm{Z}}) \times 2n}.
    \end{align}
    Conditions~\eqref{eq:CSS_eq_1a} and~\eqref{eq:CSS_eq_1b} hold, so the code is non-genuine: it collapses into the \ac{CSS} code $\mathcal{C}(C_{\mathrm{X}}, C_{\mathrm{Z}} \cap C_{\mathrm{X}})$. 
   More generally, whenever any two of $C_{\mathrm{X}}, C_{\mathrm{Y}}, C_{\mathrm{Z}}$ coincide, the XYZ code is non-genuine under uniform Pauli relabeling. 
   If the coinciding pair involves $C_{\mathrm{Y}}$, then the code is already \ac{CSS} in the original Pauli labeling. 
   Hence, using three \emph{distinct} classical codes is necessary, but not sufficient, for having a genuine XYZ code.
\end{Exa}

\begin{Exa}[4-qubit XYZ Code equivalent to a \ac{CSS} Code under Local Pauli Relabeling]
\label{ex:n4-counterexample}
    Consider these stabilizers for a XYZ code with  $n=4$ qubits:
    \begin{align}
      \6S_1 = \6X_1 \6X_4,\
      \6S_2 = \6Y_1 \6Y_3 \6Y_4,\
      \6S_3 = \6Z_1 \6Z_2 \6Z_4.
    \end{align}
    Here the underlying classical codes $\CP{X} = \vspan[1\,0\,0\,1]$, $\CP{Y} = \vspan [1\,0\,1\,1]$, $\CP{Z} = \vspan[1\,1\,0\,1]$ are distinct one-dimensional spaces with trivial pairwise intersections, so this code is genuine under uniform Pauli relabeling. 
    However, under local Pauli relabeling it becomes \ac{CSS}: transposing the $\sym{Y}$ and $\sym{Z}$ labels on qubit~$3$ alone maps $\6S_2 \to \widetilde{\6S}_2 = \6Y_1 \6Z_3 \6Y_4$, and replacing it with the product $\6S_1\,\widetilde{\6S}_2 = -\6Z_1 \6Z_3 \6Z_4$ (a row addition in the binary representation) results in the stabilizer group  generators 
    $\{ \6X_1 \6X_4, \6Z_1 \6Z_3 \6Z_4, -\6Z_1 \6Z_2 \6Z_4 \}$, 
    which are \ac{CSS}.
\end{Exa}

\begin{Exa}[A Small Genuine XYZ Code under Local Pauli Relabeling]
\label{ex:n4-genuine}
Consider an XYZ code with $n = 4$ qubits:
\begin{align}
  \6S_1 = \6X_1 \6X_3 \6X_4,\
  \6S_2 = \6Y_1 \6Y_2 \6Y_4,\
  \6S_3 = \6Z_1 \6Z_2 \6Z_3.
\end{align}
The three rows are linearly independent over $\F_2$, so none of the conditions~\eqref{eq:CSS_eq_1} holds. Moreover, an exhaustive search over all $6^4 = 1296$ assignments of single-qubit Pauli relabelings shows that no local relabeling brings the code into a \ac{CSS} form.
\end{Exa}

\subsection{Complexity of LC-equivalence to CSS Codes}

A natural question raised by the above examples is whether it is possible to efficiently determine if a given stabilizer code is \ac{LC}-equivalent to a CSS code as per Definition~\ref{def:genuine_XYZ_local}. To the best of our knowledge, the computational complexity of this decision problem is open. It is known that deciding the pairwise \ac{LC}-equivalence of stabilizer states (that is, stabilizer codes with null rate, $k=0$) can be done efficiently~\cite{bouchet1991efficient}, \cite{van2004efficient}, but we are not aware of either a polynomial-time algorithm or a hardness result for deciding whether a stabilizer code is \ac{LC}-equivalent to some \ac{CSS} code.
Related hard problems on graph states (such as the vertex-minor problem~\cite{dahlberg2022complexity}) do not seem to be directly applicable.

In conclusion, the presence of $X$-, $Y$-, and $Z$-type generators in the \ac{PCM} is not sufficient to guarantee a genuinely non-\ac{CSS} code. 
The rank criteria derived above identify when the additional $Y$-type constraints are essential. 
We next study how these additional constraints affect the minimum distance.

\section{Minimum Distance Bounds for XYZ codes}
\label{sec:bounds_dmin}

In this section, we derive upper and lower bounds on the quantum minimum distance of XYZ codes. The upper bound follows by restricting the logical-operator search to pure Pauli types. Lower bounds are obtained from \ac{CSS} codes whose stabilizer groups are contained in the XYZ stabilizer group, and are subsequently refined for mixed logical operators using the weight identity in~\eqref{eq:weight_identity}.

We begin by deriving an upper bound by studying $d^\sigma$, the minimum-weight logical errors consisting solely of $\sym{\sigma}$-type operators for $\sigma\in\{\sym{X},\sym{Y},\sym{Z}\}$.\footnote{If the set of pure $\sigma$-type logical operators is empty, we set $d^\sigma=\infty$.}
This yields the following result.

\begin{Prop}[$\sym{X}$-type Minimum-weight Logical Errors]
\label{the:X_errors}
    Let $\mathcal{C}_{\mathrm{XYZ}}(C_{\mathrm{X}},C_{\mathrm{Y}},C_{\mathrm{Z}})$ be an XYZ code.
    The minimum-weight $\sym{X}$-type logical errors have weight
    \begin{align} \label{eq:X_dist_upper}
        d^{\mathrm{X}}
        = \min\bigl\{\wt{\6v}\,\big|\,
        \6v\in (C_{\mathrm{Y}}\cap C_{\mathrm{Z}})
        \setminus(C_{\mathrm{X}}^{\perp}+C_\cap^{\perp})\bigr\}.
    \end{align}
\end{Prop}

\begin{IEEEproof}
    Restricting~\eqref{eq:stab_distance} to $\sym{X}$-type errors we have
    \begin{align}
        d^{\mathrm{X}} = 
        \min \left\{
        \wt{\6x} \, \middle| \,
        \6H_{\mathrm{XYZ}} \vstack{\60}{\6x} =\60,
        \vstack{\6x}{\60} \notin \rs(\6H_{\mathrm{XYZ}})
        \right\}.
    \end{align}
    Using the definition~\eqref{eq:PCM_XYZ}, the syndrome condition
    $\6H_{\mathrm{XYZ}} \big[\begin{smallmatrix} \60 \\ \6x \end{smallmatrix}\big] = \60$ is equivalent to  $\6H_{\mathrm{Z}}\6x=\60$ and $\6H_{\mathrm{Y}}\6x=\60$, so that $\6x\in C_{\mathrm{Y}}\cap C_{\mathrm{Z}}$. 
    Using the decomposition~\eqref{eq:red_irred} the remaining condition $\big[\begin{smallmatrix} \6x \\ \60 \end{smallmatrix}\big] \notin\rs(\6H_{\mathrm{XYZ}})$ is equivalent to $\6x \notin \rs(\6H_{\mathrm{X}}') + \rs(\6H_\cap)$, since by construction $\6H_{\mathrm{Y}}'$, $\6H_{\mathrm{Z}}'$ are irreducible components that cannot be turned into $\sym{X}$-type checks. Using  $\rs(\6H_{\mathrm{X}}') + \rs(\6H_\cap) = \rs(\6H_{\mathrm{X}}) + \rs(\6H_\cap)$, the second condition is equivalent to $\6x \notin C_{\mathrm{X}}^{\perp} + C_\cap^{\perp}$, which then gives~\eqref{eq:X_dist_upper}.
\end{IEEEproof}

Analogous results hold for $d^{\mathrm{Y}}$ and $d^{\mathrm{Z}}$. 
Summarizing,
\begin{subequations}
\begin{align}
    d^{\mathrm{X}} & = \min\bigl\{\wt{\6v}\,\big|\,
    \6v\in(C_{\mathrm{Y}}\cap C_{\mathrm{Z}})\setminus(C_{\mathrm{X}}^{\perp}+C_\cap^{\perp})\bigr\}, \\
    d^{\mathrm{Y}} & = \min\bigl\{\wt{\6v}\,\big|\,
    \6v\in(C_{\mathrm{X}}\cap C_{\mathrm{Z}})\setminus(C_{\mathrm{Y}}^{\perp}+C_\cap^{\perp})\bigr\}, \\
    d^{\mathrm{Z}} & = \min\bigl\{\wt{\6v}\,\big|\,
    \6v\in(C_{\mathrm{X}}\cap C_{\mathrm{Y}})\setminus(C_{\mathrm{Z}}^{\perp}+C_\cap^{\perp})\bigr\},
\end{align}
\end{subequations}
are the minimum-weights of logical errors consisting solely of $\sym{X}$-, $\sym{Y}$- and $\sym{Z}$-type operators, respectively. 
Since each $d^{\sigma}$ is the minimum weight of the set of logical operators restricted to a
single Pauli type, the upper bound follows.

\begin{Cor}[Distance Upper Bound]
\label{cor:upper_bound}
    Let $\mathcal{C}_{\mathrm{XYZ}}(C_{\mathrm{X}},C_{\mathrm{Y}},C_{\mathrm{Z}})$ be an XYZ code. Then
    \begin{align} \label{eq:upper}
        d \leq 
        U := 
        \min \left\{ d^{\mathrm{X}}, d^{\mathrm{Y}}, d^{\mathrm{Z}} \right\}.
    \end{align}
\end{Cor}

We now derive a lower bound on the minimum distance by considering three \ac{CSS} codes whose codespaces contain the XYZ codespace. These are obtained by removing some stabilizer generators and exploiting the decomposition~\eqref{eq:red_irred} to obtain a \ac{CSS} code. 

\begin{Def}[\ac{CSS} Codes imposed by the XYZ Structure]
\label{def:CSS_imposed}
    Let $\mathcal{C}_{\mathrm{XYZ}}(C_{\mathrm{X}},C_{\mathrm{Y}},C_{\mathrm{Z}})$ be an XYZ code, and set
    \begin{align}
        \tilde{\6H}_{\mathrm{X}} = \begin{bmatrix}\6H_{\mathrm{X}}\\\6H_\cap\end{bmatrix},
        \qquad
        \tilde{\6H}_{\mathrm{Z}} = \begin{bmatrix}\6H_{\mathrm{Z}}\\\6H_\cap\end{bmatrix},
    \end{align}
    defining the codes
    \begin{subequations}
    \begin{align}
        \tilde{C}_{\mathrm{X}} 
        & = \ker(\tilde{\6H}_{\mathrm{X}}) = C_{\mathrm{X}}\cap C_\cap,  &
        \tilde{C}_{\mathrm{X}}^\perp 
        & = C_{\mathrm{X}}^\perp + C_\cap^\perp,  
        \\
        \tilde{C}_{\mathrm{Z}} 
        & = 
        \ker(\tilde{\6H}_{\mathrm{Z}}) = C_{\mathrm{Z}}\cap C_\cap,  &
        \tilde{C}_{\mathrm{Z}}^\perp 
        & = C_{\mathrm{Z}}^\perp + C_\cap^\perp, 
    \end{align}        
    \end{subequations}
    whose orthogonality is a consequence of~\eqref{eq:prod_XYZ} and of~\eqref{eq:red_irred}. The associated \ac{CSS} code $\mathcal{C}(\tilde{C}_{\mathrm{X}},\tilde{C}_{\mathrm{Z}})$ has \ac{PCM}
    \begin{align}
    \label{eq:CSS_tilde}
        \tilde{\6H}_{\mathrm{CSS}} =
        \begin{bmatrix} [c|c]
            \tilde{\6H}_{\mathrm{X}} & \60 \\
            \60 & \tilde{\6H}_{\mathrm{Z}}
        \end{bmatrix},
    \end{align}
    and $\rs(\tilde{\6H}_{\mathrm{CSS}})$ generates exactly the $\sym{X}$- and $\sym{Z}$-type stabilizers of $\mathcal{C}_{\mathrm{XYZ}}(C_{\mathrm{X}}, C_{\mathrm{Y}}, C_{\mathrm{Z}})$.
    When $\6H_\cap$ is empty, this reduces to the \ac{PCM} of $\mathcal{C}(C_{\mathrm{X}}, C_{\mathrm{Z}})$, as in~\eqref{eq:CSS_PCM}.

    We denote the associated classical and quantum minimum distances as: 
    \begin{subequations}
    \begin{align}
        \tilde{\delta}_{\mathrm{X}} 
        & := d(\tilde{C}_{\mathrm{X}}), 
        & 
        \tilde{d}_{\mathrm{X} \setminus \mathrm{Z}} 
        & := d(\tilde{C}_{\mathrm{X}} \setminus \tilde{C}_{\mathrm{Z}}^\perp), 
        \\
        \tilde{\delta}_{\mathrm{Z}} 
        & := d(\tilde{C}_{\mathrm{Z}}) 
        & 
        \tilde{d}_{\mathrm{Z} \setminus \mathrm{X}} 
        & := d(\tilde{C}_{\mathrm{Z}} \setminus \tilde{C}_{\mathrm{X}}^\perp)
    \end{align}
    \end{subequations}
    obtained using Definition~\ref{def:d_coset}. We also define:
    \begin{align} \label{eq:quantum_d_min_tilde}
        \tilde{\delta}_{\mathrm{X},\mathrm{Z}} := \min\{\tilde{\delta}_{\mathrm{X}},\tilde{\delta}_{\mathrm{Z}}\},
        \quad
        \tilde{d}_{\mathrm{X},\mathrm{Z}} := \min\{\tilde{d}_{\mathrm{Z}\setminus\mathrm{X}},\tilde{d}_{\mathrm{X}\setminus\mathrm{Z}}\}.
    \end{align}    
\end{Def}

By applying the same construction after a uniform Pauli relabeling, we also use the analogous notation for the imposed \ac{CSS} codes $\mathcal{C}(\tilde{C}_{\mathrm{X}},\tilde{C}_{\mathrm{Y}})$ and $\mathcal{C}(\tilde{C}_{\mathrm{Y}},\tilde{C}_{\mathrm{Z}})$. In particular, we set
$\tilde{\delta}_{\mathrm{Y}} := d(\tilde{C}_{\mathrm{Y}})$, and, for every ordered pair of distinct Pauli labels $\sigma,\sigma'\in\{\mathrm{X},\mathrm{Y},\mathrm{Z}\}$, we define $\tilde{d}_{\sigma\setminus\sigma'}:=d(\tilde{C}_{\sigma}\setminus\tilde{C}_{\sigma'}^{\perp})$ and $\tilde{d}_{\sigma,\sigma'}:=\min\{\tilde{d}_{\sigma\setminus\sigma'},\tilde{d}_{\sigma'\setminus\sigma}\}$.

The specialization of Proposition~\ref{prop:CSS_d_min} to the \ac{CSS} code of Definition \ref{def:CSS_imposed} directly results in the following propositions.

\begin{Prop}[Coset Distances]
\label{cor:coset_distances}
    The distances of the \ac{CSS} codes $\mathcal{C}(\tilde{C}_{\mathrm{X}},\tilde{C}_{\mathrm{Z}})$ and $\mathcal{C}(C_{\mathrm{X}},C_{\mathrm{Z}})$ satisfy
    \begin{subequations}
    \begin{align}
        \tilde{\delta}_{\mathrm{Z}} & \geq \delta_{\mathrm{Z}},
        & 
        \tilde{d}_{\mathrm{Z}\setminus\mathrm{X}} & \geq d_{\mathrm{Z}\setminus\mathrm{X}},
        \\
        \tilde{\delta}_{\mathrm{X}} & \geq \delta_{\mathrm{X}}, 
        & 
        \tilde{d}_{\mathrm{X}\setminus\mathrm{Z}} & \geq d_{\mathrm{X}\setminus\mathrm{Z}}.
    \end{align}
    \end{subequations}
\end{Prop}

\begin{IEEEproof}
    By construction $\tilde{C}_{\mathrm{Z}} = C_{\mathrm{Z}}\cap C_\cap\subseteq C_{\mathrm{Z}}$, so $\tilde{C}_{\mathrm{Z}}$ is a subcode of $C_{\mathrm{Z}}$, 
    which proves $ \tilde{\delta}_{\mathrm{Z}} \geq \delta_{\mathrm{Z}}$, and similarly for $\tilde{\delta}_{\mathrm{X}} \geq \delta_{\mathrm{X}}$. Furthermore, $\tilde{C}_{\mathrm{X}}^\perp = C_{\mathrm{X}}^\perp + C_\cap^\perp \supseteq C_{\mathrm{X}}^\perp$, and similarly for $\tilde{C}_{\mathrm{Z}}^\perp$, so the same arguments applied to the coset distances proves $\tilde{d}_{\mathrm{Z}\setminus\mathrm{X}} \geq d_{\mathrm{Z}\setminus\mathrm{X}}$ and $\tilde{d}_{\mathrm{X}\setminus\mathrm{Z}} \geq d_{\mathrm{X}\setminus\mathrm{Z}}$.
\end{IEEEproof}

\begin{Prop}[\ac{CSS} Distance and XYZ Distance]
\label{pro:CSS_XYZ_dist}
    Let $\mathcal{C}_{\mathrm{XYZ}}(C_{\mathrm{X}},C_{\mathrm{Y}},C_{\mathrm{Z}})$ be an XYZ code and $\mathcal{C} (\tilde{C}_{\mathrm{X}}, \tilde{C}_{\mathrm{Z}})$ the derived \ac{CSS} code as per Definition~\ref{def:CSS_imposed}.
    Then $\tilde{d}_{\mathrm{X},\mathrm{Z}}\leq d$.
\end{Prop}

\begin{IEEEproof}
    The minimum distance~\eqref{eq:stab_distance_q} of $\mathcal{C} (\tilde{C}_{\mathrm{X}}, \tilde{C}_{\mathrm{Z}})$ is
    \begin{align}
        \tilde{d}_{\mathrm{X},\mathrm{Z}} =
        \min\! \left\{ \wt{\6u \vee \6v} \, \middle| \, \tilde{\6H}_{\mathrm{CSS}}\! \vstack{\6v}{\6u} = \60, \vstack{\6u}{\6v} \notin \rs(\tilde{\6H}_{\mathrm{CSS}})
        \right\}.
    \end{align}
    Since $\rs(\tilde{\6H}_{\mathrm{CSS}}) \subseteq \rs(\6H_{\mathrm{XYZ}})$, $\6H_{\mathrm{XYZ}} \big[\begin{smallmatrix}\6v\\\6u\end{smallmatrix}\big] = \60$ implies $\tilde{\6H}_{\mathrm{CSS}} \big[ \begin{smallmatrix}\6v\\\6u\end{smallmatrix} \big] = \60$ and, similarly, $\big[ \begin{smallmatrix}\6u\\\6v\end{smallmatrix} \big] \notin \rs(\tilde{\6H}_{\mathrm{XYZ}})$ implies $\big[ \begin{smallmatrix}\6u\\\6v\end{smallmatrix} \big] \notin \rs(\tilde{\6H}_{\mathrm{CSS}})$. Hence, every XYZ logical operator is also a logical operator of the \ac{CSS} code that enters the minimization defining $\tilde{d}_{\mathrm{X},\mathrm{Z}}$.
    Thus we conclude
    $\tilde{d}_{\mathrm{X},\mathrm{Z}}\le d$.
\end{IEEEproof}

The same derivation applies to $\mathcal{C}(\tilde{C}_{\mathrm{X}},\tilde{C}_{\mathrm{Y}})$
and $\mathcal{C} (\tilde{C}_{\mathrm{Y}}, \tilde{C}_{\mathrm{Z}})$, which are \ac{CSS} codes (up to uniform Pauli relabeling) obtained from the XYZ code by removing the stabilizers associated to $\6H_\mathrm{Z}'$ and to $\6H_\mathrm{X}'$, respectively. Thus we have $\tilde{d}_{\mathrm{X},\mathrm{Y}}\leq d$ and $\tilde{d}_{\mathrm{Y},\mathrm{Z}}\leq d$, that is, every non-trivial XYZ logical operator is also a non-trivial logical operator of each of these \ac{CSS} codes. The lower bound below then follows.

\begin{Cor}[Distance Lower Bound]
\label{cor:lower_bound}
    Let $\mathcal{C}_{\mathrm{XYZ}}(C_{\mathrm{X}},C_{\mathrm{Y}},C_{\mathrm{Z}})$ be an XYZ code and 
    $\mathcal{C} (\tilde{C}_{\mathrm{X}}, \tilde{C}_{\mathrm{Y}})$, $\mathcal{C} (\tilde{C}_{\mathrm{Y}}, \tilde{C}_{\mathrm{Z}})$, and $\mathcal{C} (\tilde{C}_{\mathrm{X}}, \tilde{C}_{\mathrm{Z}})$ the imposed
    \ac{CSS} codes.
    Then:
    \begin{align} 
    \label{eq:lower}
        d \geq 
        L := 
        \max
        \left\{ \tilde{d}_{\mathrm{X},\mathrm{Y}}, \tilde{d}_{\mathrm{Y},\mathrm{Z}}, \tilde{d}_{\mathrm{X},\mathrm{Z}} \right\}.
    \end{align}
\end{Cor}

Next, we prove a second lower bound, based on the identity in~\eqref{eq:weight_identity} for mixed Pauli operators.

\begin{Lem}[Conditions for Mixed Operators]
\label{lem:mixed}
    Let $\big[\begin{smallmatrix} \6u \\ \6v \end{smallmatrix}\big]$ be associated to a nontrivial mixed logical operator of $\mathcal{C}_{\mathrm{XYZ}}(C_{\mathrm{X}},C_{\mathrm{Y}},C_{\mathrm{Z}})$, i.e., it is in $\mathcal{N}(\mathcal{S})/ \Phi \mathcal{S}$.
    Writing $\6w=\6u+\6v$, these three conditions all hold:
    \begin{enumerate}
        \item[$(i)$] $\6u \notin \tilde{C}_{\mathrm{X}}^\perp$ or $\6v \notin \tilde{C}_{\mathrm{Z}}^\perp$,
        \item[$(ii)$] $\6u \notin \tilde{C}_{\mathrm{Y}}^\perp$ or $\6w \notin \tilde{C}_{\mathrm{Z}}^\perp$,
        \item[$(iii)$] $\6v \notin \tilde{C}_{\mathrm{Y}}^\perp$ or $\6w \notin \tilde{C}_{\mathrm{X}}^\perp$.
    \end{enumerate}
\end{Lem}

\begin{IEEEproof}
    Consider the CSS code $\mathcal{C}(\tilde{C}_\mathrm{X},\tilde{C}_\mathrm{Z})$ of Definition~\ref{def:CSS_imposed}. As in Proposition~\ref{pro:CSS_XYZ_dist}, all logical operators of $\mathcal{C}_{\mathrm{XYZ}}(C_{\mathrm{X}}, C_{\mathrm{Y}}, C_{\mathrm{Z}})$ are also logical operators of $\mathcal{C}(\tilde{C}_\mathrm{X},\tilde{C}_\mathrm{Z})$. 
    Since $\big[\begin{smallmatrix} \6u \\ \6v \end{smallmatrix}\big]$ is by assumption mixed, both $\6u$ and $\6v$ must be non-zero. Moreover, it is nontrivial for the XYZ code, hence it cannot be a stabilizer of the imposed \ac{CSS} code. Therefore, we cannot have simultaneously $\6u \in \tilde{C}_{\mathrm{X}}^\perp$ and $\6v \in \tilde{C}_{\mathrm{Z}}^\perp$, i.e., condition $(i)$. 
    The conditions $(ii)$ and $(iii)$ follow identically from
    $\mathcal{C}(\tilde{C}_\mathrm{Y},\tilde{C}_\mathrm{Z})$ and $\mathcal{C}(\tilde{C}_\mathrm{X},\tilde{C}_\mathrm{Y})$, which are \ac{CSS} up to uniform Pauli relabeling.
\end{IEEEproof}

\begin{Prop}[Distance Lower Bound based on Mixed Operators]
\label{prop:aug-mixed}
    Using the notation introduced in Definition~\ref{def:CSS_imposed}, and their analogues for the remaining Pauli pairs, define:
    \begin{align}
    B_{\mathrm{X}}
    & = \tfrac{1}{2}\tilde{\delta}_{\mathrm{X}}
            + \tfrac{1}{2}\min\left\{
            \tilde{d}_{\mathrm{Z}\setminus\mathrm{Y}} + \tilde{\delta}_{\mathrm{Y}},\,
            \tilde{d}_{\mathrm{Y}\setminus\mathrm{Z}} + \tilde{\delta}_{\mathrm{Z}}
            \right\},\\
    B_{\mathrm{Y}} 
    & = \tfrac{1}{2}\tilde{\delta}_{\mathrm{Y}}
            + \tfrac{1}{2} 
            \min\left\{
            \tilde{d}_{\mathrm{Z}\setminus\mathrm{X}} + \tilde{\delta}_{\mathrm{X}},\,
            \tilde{d}_{\mathrm{X}\setminus\mathrm{Z}} + \tilde{\delta}_{\mathrm{Z}}
            \right\},\\
    B_{\mathrm{Z}} 
    & = \tfrac{1}{2}\tilde{\delta}_{\mathrm{Z}}
            + \tfrac{1}{2}
            \min\left\{ 
            \tilde{d}_{\mathrm{X}\setminus\mathrm{Y}} + \tilde{\delta}_{\mathrm{Y}},\,
            \tilde{d}_{\mathrm{Y}\setminus\mathrm{X}} + \tilde{\delta}_{\mathrm{X}} 
            \right\}.
    \end{align}
    Then every mixed-type logical operator of
    $\mathcal{C}_{\mathrm{XYZ}} (C_{\mathrm{X}},C_{\mathrm{Y}},C_{\mathrm{Z}})$, has weight
    at least $B := \max\{ B_{\mathrm{X}}, B_{\mathrm{Y}}, B_{\mathrm{Z}} \}$.
\end{Prop}

\begin{IEEEproof}
    Let $\big[\begin{smallmatrix}\6u\\\6v\end{smallmatrix}\big]$ be a vector associated with a mixed logical operator of the XYZ code; we want to show that its Pauli weight is greater than or equal to $B_{\mathrm{Y}}$. 
    Since the operator is mixed, $\6u$, $\6v$, and $\6w=\6u+\6v$ are all nonzero. Moreover, by Proposition~\ref{pro:CSS_XYZ_dist} it is a logical operator of each imposed \ac{CSS} code, so $\6u\in\tilde{C}_{\mathrm{Z}}$, $\6v\in\tilde{C}_{\mathrm{X}}$ and $\6w = \6u+\6v\in\tilde{C}_{\mathrm{Y}}$, hence $\wt{\6w} \geq \tilde{\delta}_{\mathrm{Y}}$. 
    We then use condition $(i)$ of Lemma~\ref{lem:mixed}: if $\6u\notin\tilde{C}_{\mathrm{X}}^\perp$, then $\6u \in \tilde{C}_{\mathrm{Z}} \setminus \tilde{C}_{\mathrm{X}}^\perp$ and $\6v\in\tilde{C}_{\mathrm{X}}$, giving $\wt{\6u} \geq \tilde{d}_{\mathrm{Z}\setminus\mathrm{X}}$ and $\wt{\6v} \ge \tilde{\delta}_{\mathrm{X}}$; if instead $\6v \notin \tilde{C}_{\mathrm{Z}}^\perp$, then $\6v \in \tilde{C}_{\mathrm{X}} \setminus \tilde{C}_{\mathrm{Z}}^\perp$ and $\6u \in \tilde{C}_{\mathrm{Z}}$, giving $\wt{\6v} \geq \tilde{d}_{\mathrm{X}\setminus\mathrm{Z}}$ and $\wt{\6u} \geq \tilde{\delta}_{\mathrm{Z}}$. Substituting into the weight identity~\eqref{eq:weight_identity}, $\wt{\6u\vee\6v} = \tfrac{1}{2}\big(\wt{\6u}+\wt{\6v}+\wt{\6w}\big)$ and retaining the smaller of the two cases yields
    \begin{align}
       \wt{\6u\vee\6v} 
       \geq 
       \tfrac{1}{2} 
        \min\left\{
        \tilde{d}_{\mathrm{Z}\setminus\mathrm{X}} + \tilde{\delta}_{\mathrm{X}} + \tilde{\delta}_{\mathrm{Y}},\,
        \tilde{d}_{\mathrm{X}\setminus\mathrm{Z}} + \tilde{\delta}_{\mathrm{Z}} + \tilde{\delta}_{\mathrm{Y}}
        \right\} 
        = B_{\mathrm{Y}} .
    \end{align}
    
    The bounds $B_{\mathrm{X}}$ and $B_{\mathrm{Z}}$ follow identically from conditions $(ii)$ and $(iii)$ of Lemma~\ref{lem:mixed}, using the \ac{CSS} codes $\mathcal{C}(\tilde{C}_{\mathrm{X}},\tilde{C}_{\mathrm{Y}})$ and $\mathcal{C}(\tilde{C}_{\mathrm{Y}},\tilde{C}_{\mathrm{Z}})$. 
    Since each $B_\sigma$, with $\sigma \in \{ X, Y, Z \}$, bounds every mixed-type logical operator, so does their maximum.
\end{IEEEproof}

Since each Pauli operator is either of \emph{pure} type or \emph{mixed} type, the minimum distance is also attained by a logical operator that is either of pure or mixed, which leads to the following Corollary.
\begin{Cor}[Combined Distance Lower Bound]
\label{cor:combined_lower}
    Let $\mathcal{C}_{\mathrm{XYZ}}(C_{\mathrm{X}},C_{\mathrm{Y}},C_{\mathrm{Z}})$ be an XYZ code, with $U$ the minimum distance of pure operators, as in~\eqref{eq:upper}, and $B$ the minimum distance of mixed operators, as in Proposition~\ref{prop:aug-mixed}.
    Then:
    \begin{align} 
    \label{eq:combined_lower}
        d \geq \min\bigl\{U, \, 
        B \bigr\},
    \end{align}
    and we have exactly $d = U$ when $U \leq B$.
\end{Cor}

Eq.~\eqref{eq:lower} shows that every XYZ logical operator must also be logical for each imposed \ac{CSS} code, and hence cannot be lighter than the largest imposed \ac{CSS} distance. 
Eq.~\eqref{eq:combined_lower}, instead, captures genuinely mixed logical operators and can be stronger when the three component codes have large classical distances. 
Together, these bounds explain why adding $Y$-type stabilizers can increase the minimum distance relative to the underlying \ac{CSS} constructions. 

The proposed bounds require computing minimum-weight vectors in a code while excluding a given subspace. We formulate each problem as a \ac{MILP} with a candidate vector $\mathbf{x}$, linearized parity constraints enforcing $\mathbf{x}\in\mathcal{C}$, and auxiliary syndrome variables enforcing $\mathbf{x}\notin\mathcal{D}$, where $\mathcal{D}\subseteq\mathcal{C}$. The objective is to minimize $\|\mathbf{x}\|$. When the solver certifies optimality, the resulting value is exactly $d(\mathcal{C}\setminus\mathcal{D})$. These formulations involve one $n$-bit candidate vector, whereas computing the exact quantum distance requires the two binary components of a Pauli operator. Therefore, the bounds are expected to be easier to compute.

\section{Families of XYZ Stabilizer Codes}
\label{sec:XYZ_families}

In this section, we use the XYZ construction to identify three families of quantum stabilizer codes. More precisely, we show how a topological code family already known in the literature, the XYZ$^2$ code~\cite{Srivastava2022_xyzhexagonal}, is \ac{LC}-equivalent to XYZ codes. Furthermore, we construct two novel families (and as many instances) of \ac{QLDPC} XYZ codes: one based on \ac{IS} codes~\cite{Ostrev2024classicalproduct, Ostrev_IS} and one based on \ac{QD} codes~\cite{baldelli2025quantum, baldelli2026DC, baldelli2026ISIT, kirsten_QD}.

\subsection{Topological Construction based on the XYZ$^2$ Code}
\label{subsec:XYZ2}

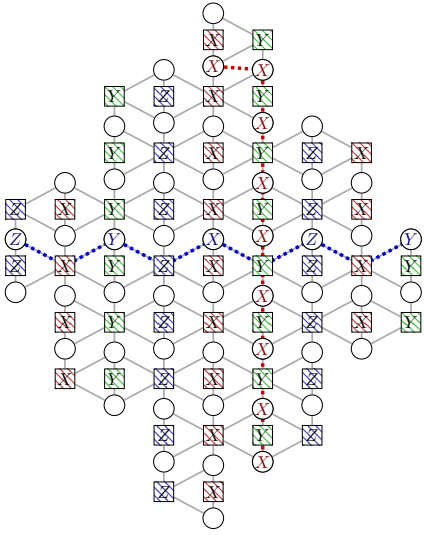
\begin{figure}[t]
    \centering
    \resizebox{0.65\columnwidth}{!}{
    \begin{tikzpicture}[
  scale=1.8,
  vn/.style ={circle, draw=black, line width=.6pt, fill=white,
              minimum size=4.2mm, inner sep=0pt},
  vnErr/.style ={circle, draw=black, line width=.6pt, fill=white,
              minimum size=4.2mm, inner sep=0pt, 
              font=\normalsize\color{red!50!black}},
  vnErr2/.style ={circle, draw=black, line width=.6pt, fill=white,
              minimum size=4.2mm, inner sep=0pt, 
              font=\normalsize\color{blue!50!black}},
  cnX/.style={rectangle, draw=black, line width=.4pt, preaction={fill=white},,
              pattern=north west lines, pattern color=red!50!gray,
              font=\normalsize, minimum size=4.0mm, inner sep=0pt},
  cnY/.style={rectangle, draw=black, line width=.5pt, preaction={fill=white},,
              pattern=north west lines, pattern color=green!50!gray,
              font=\normalsize, minimum size=4.0mm, inner sep=0pt},
  cnZ/.style={rectangle, draw=black, line width=.5pt, preaction={fill=white},,
              pattern=north west lines, pattern color=blue!50!gray,
              font=\normalsize, minimum size=4.0mm, inner sep=0pt},
  edge/.style={gray!66, line width=1.0pt},
  edgeErr/.style={dotted, red!80!black, line width=2.0pt},
  edgeErr2/.style={dotted, blue!80!black, line width=2.0pt} 
]

\pgfdeclarelayer{bg}\pgfsetlayers{bg,main}

\def\dist{5}                        
\pgfmathtruncatemacro{\dm}{\dist-1}
\pgfmathtruncatemacro{\dmm}{\dist-2}
\def\hx{0.56}                  
\def\vy{0.64}                  
\def\hh{0.30}                  

\newcommand{\CN}[4]{%
  \pgfmathtruncatemacro{\col}{mod(#1+3*\dist,3)}%
  \ifcase\col \node[cnX] (#2) at (#3,#4){$X$};%
  \or         \node[cnY] (#2) at (#3,#4){$Y$};%
  \or         \node[cnZ] (#2) at (#3,#4){$Z$};%
  \fi}

\foreach \a in {0,...,\dm}{\foreach \bb in {0,...,\dm}{
  \pgfmathsetmacro{\cx}{\hx*(\a-\bb)}
  \pgfmathsetmacro{\cy}{-\vy*(\a+\bb)}
  \node[vn] (V-\a-\bb-t) at (\cx,\cy+\hh){};
  \node[vn] (V-\a-\bb-b) at (\cx,\cy-\hh){};
}}

\foreach \a in {0,...,\dm}{\foreach \bb in {0,...,\dm}{
  \pgfmathsetmacro{\cx}{\hx*(\a-\bb)}
  \pgfmathsetmacro{\cy}{-\vy*(\a+\bb)}
  \CN{\a-\bb}{L-\a-\bb}{\cx}{\cy};
}}

\ifnum\dist>1
\foreach \a in {0,...,\dmm}{\foreach \bb in {0,...,\dmm}{
  \pgfmathsetmacro{\cx}{\hx*(\a-\bb)}
  \pgfmathsetmacro{\cy}{-\vy*(\a+\bb+1)}
  \CN{\a-\bb}{H-\a-\bb}{\cx}{\cy};
}}
\fi

\ifnum\dist>1
\foreach \a in {0,...,\dmm}{\ifodd\a\else
  \pgfmathsetmacro{\cx}{\hx*(\a+1)}
  \pgfmathsetmacro{\cy}{-\vy*\a}
  \CN{\a+1}{BUR-\a}{\cx}{\cy};
\fi}
\foreach \bb in {0,...,\dmm}{\ifodd\bb
  \pgfmathsetmacro{\cx}{-\hx*(\bb+1)}
  \pgfmathsetmacro{\cy}{-\vy*\bb}
  \CN{-\bb-1}{BUL-\bb}{\cx}{\cy};
\fi}
\foreach \bb in {0,...,\dmm}{\ifodd\bb\else
  \pgfmathsetmacro{\cx}{\hx*(\dm-\bb)}
  \pgfmathsetmacro{\cy}{-\vy*(\dm+\bb+1)}
  \CN{\dm-\bb}{BLR-\bb}{\cx}{\cy};
\fi}
\foreach \a in {0,...,\dmm}{\ifodd\a
  \pgfmathsetmacro{\cx}{\hx*(\a-\dm)}
  \pgfmathsetmacro{\cy}{-\vy*(\a+\dm+1)}
  \CN{\a-\dm}{BLL-\a}{\cx}{\cy};
\fi}
\fi


\node[vnErr] (X0) at (0,-\vy-\hh+\vy){$X$};
\foreach \a in {1,...,\dm}{
  \node[vnErr] (X\a-0) at (\hx,-2*\vy*\a+\hh+\vy){$X$};
  \node[vnErr] (X\a-1) at (\hx,-2*\vy*\a-\hh+\vy){$X$};
}

\node[vnErr2] (L0) at (-4*\hx,-5*\vy + \hh+\vy){$Z$};
\node[vnErr2] (L1) at (-2*\hx,-5*\vy + \hh+\vy){$Y$};
\node[vnErr2] (L2) at ( 0*\hx,-5*\vy + \hh+\vy){$X$};
\node[vnErr2] (L3) at ( 2*\hx,-5*\vy + \hh+\vy){$Z$};
\node[vnErr2] (L4) at ( 4*\hx,-5*\vy + \hh+\vy){$Y$};

\begin{pgfonlayer}{bg}
  \foreach \a in {0,...,\dm}{\foreach \bb in {0,...,\dm}{
     \draw[edge] (L-\a-\bb)--(V-\a-\bb-t);
     \draw[edge] (L-\a-\bb)--(V-\a-\bb-b);
  }}
  \ifnum\dist>1
  \foreach \a in {0,...,\dmm}{\foreach \bb in {0,...,\dmm}{
     \pgfmathtruncatemacro{\ap}{\a+1}\pgfmathtruncatemacro{\bp}{\bb+1}
     \draw[edge] (H-\a-\bb)--(V-\a-\bb-b);
     \draw[edge] (H-\a-\bb)--(V-\ap-\bp-t);
     \draw[edge] (H-\a-\bb)--(V-\a-\bp-t);
     \draw[edge] (H-\a-\bb)--(V-\a-\bp-b);
     \draw[edge] (H-\a-\bb)--(V-\ap-\bb-t);
     \draw[edge] (H-\a-\bb)--(V-\ap-\bb-b);
  }}
  \fi
  \ifnum\dist>1
  \foreach \a in {0,...,\dmm}{\ifodd\a\else\pgfmathtruncatemacro{\ap}{\a+1}
     \draw[edge] (BUR-\a)--(V-\a-0-t);
     \draw[edge] (BUR-\a)--(V-\a-0-b);
     \draw[edge] (BUR-\a)--(V-\ap-0-t);
  \fi}
  \foreach \a in {0,...,\dmm}{\ifodd\a\pgfmathtruncatemacro{\ap}{\a+1}
     \draw[edge] (BLL-\a)--(V-\ap-\dm-t);
     \draw[edge] (BLL-\a)--(V-\ap-\dm-b);
     \draw[edge] (BLL-\a)--(V-\a-\dm-b);
  \fi}
  \foreach \bb in {0,...,\dmm}{\ifodd\bb\pgfmathtruncatemacro{\bp}{\bb+1}
     \draw[edge] (BUL-\bb)--(V-0-\bb-t);
     \draw[edge] (BUL-\bb)--(V-0-\bb-b);
     \draw[edge] (BUL-\bb)--(V-0-\bp-t);
  \fi}
  \foreach \bb in {0,...,\dmm}{\ifodd\bb\else\pgfmathtruncatemacro{\bp}{\bb+1}
     \draw[edge] (BLR-\bb)--(V-\dm-\bp-t);
     \draw[edge] (BLR-\bb)--(V-\dm-\bp-b);
     \draw[edge] (BLR-\bb)--(V-\dm-\bb-b);
  \fi}
  \fi
  
  \draw[edgeErr] (X0)--(X1-0)--(X\dm-1);
  \draw[edgeErr2] (L0)--(-3*\hx,-5*\vy+\vy)--
                 (L1)--(-1*\hx,-5*\vy+\vy)--
                 (L2)--( 1*\hx,-5*\vy+\vy)--
                 (L3)--( 3*\hx,-5*\vy+\vy)--(L4);

\end{pgfonlayer}

\end{tikzpicture}
    }
    \caption{Tanner graph of the $\mathcal{C}_\mathrm{XYZ-2D}\llbracket 50,1,5 \rrbracket$. Circles represent qubits and squares represent stabilizers. The qubits joined by a dashed line form minimum-weight logical operators of pure type (in red, weight $2d-1=9$) and mixed type (in blue, weight $d=5$).}
    \label{fig:hexagonal_code}
\end{figure}

The XYZ$^2$ code introduced in~\cite{Srivastava2022_xyzhexagonal} is a 2D topological code family with parameters $\llbracket 2d^2, 1, d \rrbracket$, obtained by placing qubits in a rhombus-shaped patch of the hexagonal lattice. Its mixed stabilizer generators are weight-$6$ plaquette operators with type $\sym{X Y Z X Y Z}$, together with weight-$2$ links with type $\sym{X X}$ on the vertical edges, and weight-$3$ half-plaquette operators on the boundary. Via local Pauli relabeling, every stabilizer generator can be homogenized to a pure one and the resulting family $\mathcal{C}_\mathrm{XYZ-2D}\llbracket 2d^2,1,d \rrbracket$ is XYZ, in the sense of Definition~\ref{def:XYZ_codes}, where the code parameters are invariant under \ac{LC} operations and thus inherited from the XYZ$^2$ code. In particular, each plaquette becomes a pure $X$-, $Y$-, or $Z$-type stabilizer, and so does each vertical link\footnote{We conjecture that all codes in this family are genuine under uniform relabeling, which we have numerically verified for $d\leq 99$, while we have not verified genuineness under arbitrary local relabeling as it requires $6^{2d^2}$ tests, which is computationally challenging already for $d=3$.}. The Tanner graph of the $\mathcal{C}_\mathrm{XYZ-2D}$ code for $d = 5$ is shown in Fig.~\ref{fig:hexagonal_code}.

It was proven that the XYZ$^2$ codes have $d^{\mathrm{X}} = d$ and $d^{\mathrm{Y}} = d^{\mathrm{Z}}= 2d^2$, resulting in excellent resilience against $Z$- and $Y$-biased noise; however, these properties are \emph{not} \ac{LC}-invariant~\cite{Srivastava2022_xyzhexagonal}. 
Preliminary investigations indicate that the minimum-weight pure-type logical operators of $\mathcal{C}_\mathrm{XYZ-2D}$ are strings joining the two opposite corners of the rhombus (see Fig.~\ref{fig:hexagonal_code}), which would force $d^\sigma = 2d-1$. This conjecture implies that minimum-weight logical operators must be mixed and would provide an explicit example of a family for which the pure-type upper bound $U=2d-1>d$ is asymptotically loose by a factor approaching two; finally, this would endow the code with resilience against \emph{any unknown} biased noise. 
However, such an unknown bias is a rather contrived error model, and for conventional noise channels, this code is suboptimal with respect to the XYZ$^2$ code, whose pure $Y$-type and $Z$-type distances scale as $n$. 
We therefore do not analyze the $\mathcal{C}_\mathrm{XYZ-2D}$ family further.

\subsection{XYZ Codes based on Intersecting Subset Codes} \label{subsec:IS_XYZ}

\ac{IS} codes~\cite{Ostrev_IS} stand as a generalization of the classical product code construction for \ac{CSS} codes~\cite{Ostrev2024classicalproduct} and of quantum Reed-Muller codes~\cite{steane1999}. To generalize the definition \ac{IS} codes to the XYZ framework, we need to add a further \ac{PCM}, associated with $Y$-type stabilizer generators.

\begin{Def}[XYZ Code based on \ac{IS} Codes]
\label{def:XYZ_IS}
    An \ac{IS} XYZ code $\mathcal{C}_{\mathrm{XYZ}-\emph{IS}}(C_{\mathrm{X}}, C_{\mathrm{Y}}, C_{\mathrm{Z}})\llbracket n, k, d \rrbracket$, with $n=2^\ell$, is designed starting from three tuples of subsets:
    \begin{align}
    \notag
        \mathcal{S}_{\mathrm{X}} = 
        \left(\begin{array}{c}
            S_1^{\mathrm{X}} \\ 
            \vdots \\ 
            S_{a_X}^{\mathrm{X}}
        \end{array}\right), 
        & 
        \quad
        \mathcal{S}_{\mathrm{Y}} = \left(\begin{array}{c}
            S_1^{\mathrm{Y}} \\ 
            \vdots \\ 
            S_{a_Y}^{\mathrm{Y}}
            \end{array}\right), 
        \quad
        \mathcal{S}_{\mathrm{Z}} = 
        \left(\begin{array}{c}
            S_1^{\mathrm{Z}} \\ 
            \vdots \\ 
            S_{a_Z}^{\mathrm{Z}}
        \end{array}\right), \\ 
        & \quad \forall \, \sigma \in \{X, Y, Z\}, \forall \, i : \ S_i^\sigma \subseteq [\ell],
    \end{align}
    where $S_i^\sigma$, for $i=1,\ldots,a_\sigma$ with $a_\sigma < 2^\ell$, are distinct subsets of $[\ell]$ that satisfy the following intersection conditions:
    \begin{align} 
    \label{eq:int_cond}
        \forall \, \sigma \neq \sigma': \forall \, i\in[a_\sigma], j\in[a_{\sigma'}]: \quad 
        S_i^{\sigma} \cap S_j^{\sigma'} \neq \emptyset .
    \end{align}
    
    The \acp{PCM} of the XYZ code are given by:
    \begin{align}
        \6H_{\mathrm{X}} =
        \begin{bmatrix}
            \6H_{S^{\mathrm{X}}_1} \\ 
            \vdots \\ 
            \6H_{S^{\mathrm{X}}_{a_{\mathrm{X}}}}
        \end{bmatrix}\!,
        \, \, 
        \6H_{\mathrm{Y}} =
        \begin{bmatrix}
            \6H_{S^{\mathrm{Y}}_1} \\ 
            \vdots \\ 
            \6H_{S^{\mathrm{Y}}_{a_{\mathrm{Y}}}}
        \end{bmatrix}\!,
        \, \, 
        \6H_{\mathrm{Z}} =
        \begin{bmatrix}
            \6H_{S^{\mathrm{Z}}_1} \\ 
            \vdots \\ 
            \6H_{S^{\mathrm{Z}}_{a_{\mathrm{Z}}}}
        \end{bmatrix}\!,
    \end{align}
    where we have
        \begin{align}
            & \6H_{S} = \bigotimes_{i=1}^\ell
        \begin{cases}
        	  \6h & \text{if} \ i \in S \\ 
            \6I_2 & \text{if} \ i \notin S
        \end{cases}
        \quad 
        \forall S \subseteq [\ell],
        \end{align}
        with $\6h = [1~1]$.
\end{Def}
In particular, the vector $\6h$ satisfies $\6h \6h^{\top} 
= 0$ which, together with the intersecting subset condition~\eqref{eq:int_cond}, guarantees that the three \acp{PCM} 
$\6H_{\mathrm{X}}$, $\6H_{\mathrm{Y}}$, $\6H_{\mathrm{Z}}$ are pairwise 
orthogonal\footnote{Both XYZ and \ac{CSS} codes based on \ac{IS} codes can also be constructed from arbitrary tuples of matrices $\6H_i^\sigma$ such that $\6H_i^\sigma (\6H_i^{\sigma'})^\top = \60$ for all $i\in[\ell]$ and $\sigma \neq \sigma'$, instead of using the fixed base matrix $\6h = [1~1]$, see~\cite{Ostrev_IS}.}.

The original \ac{CSS} version of the \ac{IS} codes are shown to be quantum \ac{GRM} codes, also known as decreasing monomial codes~\cite{bardet2016algebraic}; using this correspondence, it is possible to determine the quantum (coset) minimum distance in $\mathrm{poly}(n)$ time~\cite{Ostrev_IS}. The identification with \ac{GRM} codes does not hold for XYZ-\ac{IS} codes, and we could not find other polynomial-time algorithms to compute the minimum distance. However, it is possible to efficiently evaluate the bounds of Section~\ref{sec:bounds_dmin} by exploiting the closure of \ac{CSS} \ac{IS} codes under sum, intersection, and orthogonal complements. These closures can be shown by representing the codes $C_\sigma^\perp = \rs(\6H_{\sigma})$ with the tuple $\mathcal{S}_\sigma$ and through this with the sets $\mathcal{F}_\sigma := \mathcal{G}(C_\sigma^\perp) = \mathcal{F}(\mathcal{S}_\sigma)$, where
\begin{align}
    \mathcal{F}(\mathcal S) := \{ F \subseteq[\ell] \mid F \cap S = \varnothing \text{ for some } S \in \mathcal S \} .
\end{align}
The correspondence is given explicitly by:
\begin{align}
    C_\sigma^\perp = \rs\!
    \left( 
    \Pi[\mathcal{F}_\sigma]  
    \begin{bmatrix}
        1 & 1\\
        0 & 1
    \end{bmatrix}^{\otimes \ell}
    \right)
\end{align}
where $\Pi[\mathcal{F}_\sigma] \in \mathbb{F}_2^{|\mathcal{F}_\sigma| \times 2^\ell}$ is a matrix with a one in each row, in the position indexed by the corresponding element of $\mathcal{F}_\sigma$.
Sums and intersections of the codes $C_\sigma^\perp$ correspond, respectively, to unions and intersections of the associated $\mathcal{F}_\sigma$, while the dual code $C_\sigma$ corresponds to $\mathcal{G}(C_\sigma)=\{ F^c \mid F \in 2^{[\ell]} \setminus \mathcal{G}(C_\sigma^\perp)\}$, where $(\cdot)^c$ denotes set complement in $[\ell]$. Since these families are downward closed with the partial order given by set inclusion, a generating tuple for a set $\mathcal{F}$ is recovered via $\mathcal{S}(\mathcal{F}) = \{ F^{c} \mid F \in \max(\mathcal{F}) \}$.
Hence, each $C_\sigma^\perp$, their duals $C_\sigma$, the pairwise intersections $C_\sigma\cap C_{\sigma'}$, the reducible component $C_\cap^{\perp}$, and the imposed codes $\tilde{C}_{\mathrm{X}},\tilde{C}_{\mathrm{Y}},\tilde{C}_{\mathrm{Z}}$ of Definition~\ref{def:CSS_imposed}, are all \ac{IS} codes whose $\mathcal{F}$ sets are computed explicitly from $\mathcal{S}_{\mathrm{X}}, \mathcal{S}_{\mathrm{Y}}, \mathcal{S}_{\mathrm{Z}}$. E.g., we have
\begin{align}
    \mathcal{G}(C_\cap^\perp) = (\mathcal F_{\mathrm X}\cap\mathcal F_{\mathrm Y})\cup(\mathcal F_{\mathrm Y}\cap\mathcal F_{\mathrm Z})\cup(\mathcal F_{\mathrm X}\cap\mathcal F_{\mathrm Z}).
\end{align}
Therefore, since $\tilde C_{\mathrm X}^{\perp}=C_{\mathrm X}^{\perp}+C_\cap^\perp$, we obtain
\begin{align}
    \mathcal{G}(\tilde C_{\mathrm X}^{\perp})=\mathcal F_{\mathrm X}\cup\bigl(\mathcal F_{\mathrm Y}\cap\mathcal F_{\mathrm Z}\bigr).
\end{align}

As an example, we construct an XYZ code $\mathcal{C}_{\mathrm{XYZ}-\text{IS}}\llbracket 512, 9, d \rrbracket$, with $12\leq d \leq 16$, choosing $\ell = 9$ and the following three tuples 
\begin{align}
\begin{aligned}
\mathcal{S}_{\mathrm{X}} &\!=\!
\left(
\begin{array}{@{}c@{}}
\{1,2,3\}\\
\{4,5,6\}\\
\{7,8,9\}
\end{array}
\right),
\;
\mathcal{S}_{\mathrm{Y}} \!=\!
\left(
\begin{array}{@{}c@{}}
\{1,5,9\}\\
\{2,6,7\}\\
\{3,4,8\}\\
\{3,5,7\}\\
\{2,4,9\}\\
\{1,6,8\}
\end{array}
\right),
\;
\mathcal{S}_{\mathrm{Z}} \!=\!
\left(
\begin{array}{@{}c@{}}
\{1,4,7\}\\
\{2,5,8\}\\
\{3,6,9\}
\end{array}
\right),
\end{aligned}
\end{align}
satisfying the intersection condition given in~\eqref{eq:int_cond}.
We have verified that this code is genuine, also under uniform relabeling, while extending this verification to arbitrary local relabelings would require evaluating $6^{512}$ cases, which is computationally infeasible.   
We have bounded the quantum minimum distance using the closed-form distance equations of~\cite{Ostrev_IS} for \ac{CSS} \ac{IS} codes to compute the bounds of Section~\ref{sec:bounds_dmin}.
In particular, the distances with pure-type operators are $d^\mathrm{X} = d^\mathrm{Z} = 32$ and $d^\mathrm{Y} = 16$, giving $U = 16$ from Corollary~\ref{cor:upper_bound}.
The three \ac{CSS} codes of Definition~\ref{def:CSS_imposed} have $\tilde{\delta}_\sigma = 8$ and $\tilde{d}_{\sigma \setminus \sigma'} = 8$ for all $\sigma \neq \sigma'$ with $\sigma, \sigma' \in \{ X, Y, Z \}$, hence $L = 8$ from Corollary~\ref{cor:lower_bound}; finally, Proposition~\ref{prop:aug-mixed} yields $B_\mathrm{X} = B_\mathrm{Y} = B_\mathrm{Z} = 12$, so that Corollary~\ref{cor:combined_lower} gives the final result $12 \leq d \leq 16$. This is a concrete instance in which the bound~\eqref{eq:combined_lower} is strictly stronger than the one in~\eqref{eq:lower}.

\subsection{XYZ Codes based on Quasi-Dyadic Codes} 
\label{subsec:QD_XYZ}

We first recall the definitions of dyadic and \ac{QD} matrices.

\begin{Def} [Ring of Dyadic Matrices]
    Consider $\ell \in \mathbb{N}$. 
    We define $\mathcal M_{\ell}(\mathbb F_2)$
    as the set of $2^\ell\times 2^\ell$ matrices with entries over $\mathbb F_2$ and structured as follows
    \begin{align}
        \6M = 
        \begin{bmatrix}
        \6A & \6B \\
        \6B & \6A
        \end{bmatrix}, \quad \6A,\6B\in\mathcal M_{\ell-1}(\mathbb F_2).
    \end{align}
    For $\ell = 0$, $\mathcal M_0(\mathbb F_2) := \mathbb{F}_2$.
\end{Def}

Dyadic matrices are \emph{reproducible matrices}~\cite{santini_iacr}, i.e., it is possible to determine the entire matrix only using its first row, called \emph{signature}.

\begin{Def} [Dyadic Permutation Matrix, Quasi-Dyadic Matrix and Codes]
    We call \emph{\ac{DPM}} any dyadic matrix with a signature of weight $1$. 
    A matrix is called \ac{QD} if it is a $u \times w$ array of dyadic matrices; consequently, a classical linear code is called \ac{QD} if it admits a generator matrix or a \ac{PCM} that is \ac{QD}. 
\end{Def}

We construct \ac{CSS} \ac{QLDPC} codes using the \emph{affine construction} method of~\cite{baldelli2026ISIT}, that allows to have \ac{CSS} codes suitable for the CAMEL framework~\cite{CAMEL_Sisi}. 
Let us briefly recall how to build the binary \ac{PCM} of the component \ac{QD} codes.
A classical \ac{QD} code is described by an exponent matrix $\6E' \in \mathbb{F}_{2^\ell}^{u \times 2^{\ell}}$, with $u \le 2^\ell$. Each entry $p_{i,j}$ of the exponent matrix, where $i,j \in \mathbb F_{2^\ell}$, is associated to a (not necessarily distinct) \ac{DPM}.
Each row of $\6E'$ is generated as an affine permutation $p_{i,j} = a_i j + b_i$, with distinct non-zero multipliers $a_i \in \mathbb F_{2^\ell}^{\times}$ and $b_i \in \mathbb F_{2^\ell}$ with a multiplication  over the field $\mathbb F_{2^\ell}$.
Then, we \emph{lift} the exponent matrix, i.e., we replace each element $p_{i,j} \in \mathbb{F}_{2^{\ell}}$ with the corresponding \ac{DPM}.

\begin{Def} [XYZ Code based on \ac{QD} Codes]
An XYZ code $\mathcal{C}_{\mathrm{XYZ}-\emph{QD}}(C_{\mathrm{X}}, C_{\mathrm{Y}}, C_{\mathrm{Z}})\llbracket n, k, d \rrbracket$ based on \ac{QD} codes is defined by three exponent matrices $\6E'_{\mathrm{X}} \in \mathbb{F}_{2^\ell}^{u_{\mathrm{X}} \times  2^{\ell}}$, $\6E'_{\mathrm{Y}} \in \mathbb{F}_{2^\ell}^{u_{\mathrm{Y}} \times  2^{\ell}}$, $\6E'_{\mathrm{Z}}\in \mathbb{F}_{2^\ell}^{u_{\mathrm{Z}} \times  2^{\ell}}$ with affine rows
\begin{align}
    p^{(\sigma)}_{i,j} = a^{(\sigma)}_i j + b^{(\sigma)}_i, \qquad \quad \sigma \in\{X, Y, Z\}, 
\end{align}
where $u_{\mathrm{X}} + u_{\mathrm{Y}} + u_{\mathrm{Z}} < 2^{\ell}$, and the multipliers $a_i$ are distinct within and across the three matrices. 
Lifting the exponent matrices yields $\6H_{\mathrm{X}}' \in \mathbb{F}_{2}^{u_{\mathrm{X}}2^{\ell} \times 2^{2\ell}}$, $\6H_{\mathrm{Y}}' \in \mathbb{F}_{2}^{u_{\mathrm{Y}}2^{\ell} \times 2^{2\ell}}$ and $\6H_{\mathrm{Z}}' \in \mathbb{F}_{2}^{u_{\mathrm{Z}}2^{\ell} \times 2^{2\ell}}$, that satisfy
\begin{align}
    \6H'_{\sigma} (\6H'_{\sigma'})^\top = \6{1}_{u_\sigma 2^\ell \times u_{\sigma'} 2^\ell}, \, \, \, \, \sigma, \sigma' \in\{X, Y, Z\}, \, \sigma \neq \sigma'.
\end{align}
Finally, by setting
\begin{align}
    \6{H}_{\mathrm{X}} = 
    \begin{bmatrix} [c|c]
        \6{H}'_{\mathrm{X}} & \6{1}
    \end{bmatrix},
    \,
    \6{H}_{\mathrm{Y}} = 
    \begin{bmatrix} [c|c]
        \6{H}'_{\mathrm{Y}} & \6{1}
    \end{bmatrix},
    \,  
    \6{H}_{\mathrm{Z}} = 
    \begin{bmatrix} [c|c]
        \6{H}'_{\mathrm{Z}} & \6{1}
    \end{bmatrix},
\end{align}
we get the orthogonality condition
\begin{align}
    \6H_{\sigma} \6H_{\sigma'}^\top = \6{0} \qquad \sigma, \sigma' \in\{X, Y, Z\}, \, \, \sigma \neq \sigma'.
\end{align}
\end{Def}

As an example, the XYZ code $\mathcal{C}_{\mathrm{XYZ}-\text{QD}} \llbracket 257, 116, 16 \rrbracket$\footnote{The minimum distance was numerically estimated by keeping track of the minimum weight logical errors during the Monte Carlo simulations for the \ac{LER} under \ac{BP} decoding.}
$\ell = 4$, $2^{\ell} = 16$, $u_{\mathrm{X}} = u_{\mathrm{Y}} = u_{\mathrm{Z}} = 5$.
We have verified that this code instance is genuine and remains genuine under uniform Pauli relabeling.
As in Section~\ref{subsec:IS_XYZ}, we avoid the verification for the local Pauli relabeling due to complexity constraints.

Regarding the minimum-distance bounds, we use the tool available in~\cite{Pryadko_2022} to numerically estimate the distances $\tilde{d}_{\mathrm{X},\mathrm{Y}}$, $\tilde{d}_{\mathrm{Y},\mathrm{Z}}$, and $\tilde{d}_{\mathrm{X},\mathrm{Z}}$, as well as the coset distances $\tilde{d}_{\sigma\setminus\sigma'}$ for all distinct $\sigma,\sigma' \in \{\sym{X},\sym{Y},\sym{Z}\}$. These computations yield the numerical estimate $L=12$, according to Corollary~\ref{cor:lower_bound}. To evaluate the mixed-operator bound (Corollary~\ref{cor:combined_lower}), we additionally estimate the classical distances $\tilde{\delta}_{\mathrm{X}}$, $\tilde{\delta}_{\mathrm{Y}}$, and $\tilde{\delta}_{\mathrm{Z}}$ using the tool available in~\cite{MacKayDist}, obtaining $B=18$. 
Finally, the \ac{MILP} formulation solved with the Gurobi optimizer finds a pure logical operator for each $\sigma$-type, with $\sigma \in \{ X, Y, Z \}$, with weight $16$. 
The numerical estimation of the upper bound is thus $U = 16$ (Corollary~\ref{cor:upper_bound}). Together, these numerical results suggest that the minimum distance is $d=16$ (Corollary~\ref{cor:combined_lower}).

\section{Numerical Results}
\label{sec:num_res}

\begin{figure*}[t]
    \centering
    \subfloat[]{
        \resizebox{0.75\columnwidth}{!}{
            \begin{tikzpicture} 
    \begin{loglogaxis}[
        xlabel={$p$},
        ylabel={ LER },
        xmin=0.01, xmax=0.15,
        ymin=0.0000001, ymax=1,
        legend style={legend columns=1},
        grid=both,
        major grid style={solid, gray!70},
        minor grid style={solid, gray!30},
        legend style = {
            at={(0.5, -0.25)},
            anchor=north,
            legend columns = 1,
            font=\scriptsize,
        },
        legend pos = south east,
        height  = 8.0cm,
        width = \columnwidth, 
        scaled x ticks=false
    ]

    \addplot[ 
        color = NavyBlue, 
        mark = square*, thick,
        mark options = {solid},
        ]
        coordinates {
            (0.20, 1)
            (0.15, 1)
            (0.14, 1)
            (0.13, 1)
            (0.12, 1)
            (0.11, 1)
            (0.1, 0.980392156862745)
            (0.09, 0.917431192660551)
            (0.08, 0.793650793650794)
            (0.07, 0.571428571428571)
            (0.06, 0.346020761245675)
            (0.05, 0.126903553299492)
            (0.04, 0.0262743037309511)
            (0.03, 0.00264284581637507)
            (0.02, 9.74e-05)
            (0.01, 4.97e-07)
        };

    \addplot[
        color = NavyBlue,
        mark = square, thick, dashed, 
        mark options = {solid},
        line width = 0.5pt
        ]
        coordinates {
        (0.20, 1)
        (0.1, 0.970667)
        (0.0666667, 0.561175)
        (0.0444444, 0.0921592)
        (0.0296296, 0.00591809)
        (0.0197531, 0.000251917)
        (0.0131687, 1.4911e-05)
    };

    \addplot[
        color = cyan, 
        mark = square*, thick,
        mark options = {solid}
        ]
        coordinates {
            (0.20, 1)
            (0.15, 1)
            (0.14, 1)
            (0.13, 1)
            (0.12, 1)
            (0.11, 1)
            (0.10, 1)
            (0.09, 1)
            (0.08, 0.961538461538462)
            (0.07, 0.847457627118644)
            (0.06, 0.72463768115942)
            (0.05, 0.411522633744856)
            (0.04, 0.193423597678917)
            (0.03, 0.0539374325782093)
            (0.02, 0.00646078304690528)
            (0.01, 0.000439860124480415)
            (0.009, 0.000360034707345788)
            (0.008, 0.00026015312613004)
            (0.007, 0.000188697046891216)
            (0.006, 0.000124265127104585)
            (0.005, 7.44951648913227e-05)
            (0.004, 4.81159831830695e-05)
        };

    \addplot[
        color = Emerald, 
        mark = square*, thick,
        mark options = {solid}
        ]
        coordinates {
            (0.20, 1)
            (0.15, 1)
            (0.14, 1)
            (0.13, 1)
            (0.12, 1)
            (0.11, 1)
            (0.1, 1)
            (0.09, 1)
            (0.08, 1)
            (0.07, 0.970873786407767)
            (0.06, 0.813008130081301)
            (0.05, 0.571428571428571)
            (0.04, 0.256410256410256)
            (0.03, 0.0736919675755343)
            (0.02, 0.00796368559369276)
            (0.01, 0.00013033323601785)
            (0.009, 7.04994391769613e-05)
            (0.008, 3.15467066342409e-05)
            (0.007, 1.67518345162413e-05)
        };
    

    \legend{
        {$\mathcal{C}_{\mathrm{XYZ}-\text{QD}}$},
        {$\mathcal{C}_{\text{QD}}$},
        {$\mathcal{C}_{\text{Bic,}1}$},
        {$\mathcal{C}_{\text{GB,}1}$},
    }
    
    \end{loglogaxis}
\end{tikzpicture}
        }
        \label{fig:QD_XYZ_SOTA}
    }
    \subfloat[]{
        \resizebox{0.75\columnwidth}{!}{
            \input{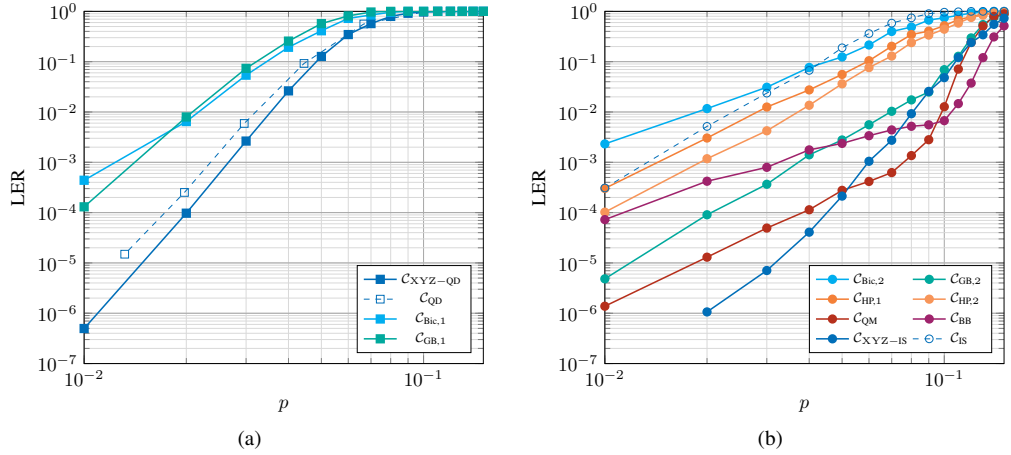}
        }
        \label{fig:IS_XYZ_SOTA}
    }
    \caption{Monte Carlo simulations of the \ac{LER} under \ac{BP4} \ac{CN} serial decoding over a code-capacity noise model, as a function of $p$. 
    (a): comparison between the \ac{LER} of $\mathcal{C}_{\mathrm{XYZ}-\text{QD}}$ and some state-of-the-art \ac{CSS} codes. With dashed lines and empty squares, we report the code $\mathcal{C}_{\text{QD}}$ decoded using the CAMEL decoder of~\cite{baldelli2026ISIT}. 
    (b): comparison between the \ac{LER} of $\mathcal{C}_{\mathrm{XYZ}-\text{IS}}$ and several state-of-the-art \ac{CSS} codes. With dashed lines and empty dots, we report the code $\mathcal{C}_{\text{IS}}$.}
    \label{fig:QD_IS_XYZ_SOTA}
\end{figure*}

In this section, we show the \ac{LER} performance of the proposed 
XYZ \ac{QLDPC} codes, one based on \ac{IS} and the other on \ac{QD} codes, through finite-length Monte Carlo simulations, over a code-capacity noise model, with depolarizing probability $p$. 
The two presented instances of XYZ codes are compared with representative state-of-the-art \ac{CSS} \ac{QLDPC} codes having comparable lengths, stabilizer generator weights, and rates.
We employ a \ac{BP4} decoder, relying on the sum-product algorithm (SPA).
We used a \ac{CN} serialized schedule for message passing.
The decoder runs at most for $50$ iterations, and the simulation continues with the same $p$ until $100$ logical errors are detected.

We compare the proposed XYZ codes with several \ac{CSS} codes, namely: a \ac{QD} code ($\mathcal{C}_{\text{QD}}$)~\cite{baldelli2026ISIT}, an
\ac{IS} code ($\mathcal{C}_{\text{IS}}$)~\cite{Ostrev2024classicalproduct}, 
bicycle codes ($\mathcal{C}_{\text{Bic}}$)~\cite{sparse_McKay}, 
\ac{GB} codes ($\mathcal{C}_{\text{GB}}$)~\cite{Kovalev_Pryadko}, 
\ac{HP} codes ($\mathcal{C}_{\text{HP}}$)~\cite{Tillich_HGP}, a
\ac{QM} code ($\mathcal{C}_{\text{QM}}$)~\cite{Pacenti_Margulis}, and a
\ac{BB} code ($\mathcal{C}_{\text{BB}}$)~\cite{Bravyi2024_BB_codes}.
For each tested code instance (be it \ac{CSS} or XYZ) the code parameters $n$, $k$, $R$, and $d$, the stabilizer generator weight $w_{\text{r}}$ (which is uniform for all the codes), and the girth $g_{\sigma}$, of each classical code $C_{\sigma}$ with $\sigma \in \{ X, Y, Z \}$, are reported in Table~\ref{tab:code_parameters}.
All tested \ac{CSS} codes are directly available in the literature, with the exception of the bicycle codes $\mathcal{C}_{\text{Bic,}1}$ and $\mathcal{C}_{\text{Bic,}2}$, constructed with the same method of~\cite{sparse_McKay}, and the \ac{GB} codes $\mathcal{C}_{\text{GB,}1}$ and $\mathcal{C}_{\text{GB,}2}$, designed as in~\cite{Hagiwara2007}\footnote{ In~\cite{Pant_Kal_almost_linear}, the authors remark that such class of \ac{CSS} codes can be considered as a specific case of \ac{GB} codes.}. 
Since instances simultaneously matched in length, dimension, stabilizer weight, are generally not available, we select the benchmark codes to match the parameters most relevant to each comparison.

In Fig.~\ref{fig:QD_XYZ_SOTA}, we compare the performance of $\mathcal{C}_{\mathrm{XYZ}-\text{QD}}$, defined in Section~\ref{subsec:QD_XYZ}, against some state-of-the-art \ac{CSS} \ac{QLDPC} codes instances with $n \approx 257$, $w_{\mathrm{r}} \approx 16$ and $R \approx 0.45$. We observe that the proposed $\mathcal{C}_{\mathrm{XYZ}-\text{QD}}$ behaves better than $\mathcal{C}_{\text{Bic,}1}$ and $\mathcal{C}_{\text{GB,}1}$, consistently with their smaller minimum distance, namely $d_{\text{Bic,}1} = 5$ and $d_{\text{GB,}1} = 8$, even though the latter exhibits a larger girth, equal to $6$. 
As a benchmark, we report also the performance of $\mathcal{C}_{\text{QD}}$ (dashed line with empty squares), which is the state-of-the-art code instance of \ac{QD} codes, for these specific length and rate. Such a code has been devised for a joint code and \ac{BP4} decoder framework, namely CAMEL~\cite{CAMEL_Sisi}, and therefore it has been decoded with the corresponding CAMEL decoder.
The latter significantly improves the code \ac{LER} with respect to a standard \ac{BP4}~\cite{baldelli2026ISIT}.
Despite this, the performance of $\mathcal{C}_{\mathrm{XYZ}-\text{QD}}$, decoded with a conventional \ac{BP4} decoder, is superior to that of $\mathcal{C}_{\text{QD}}$.  
This comparison is informative because both codes have the same length, stabilizer generator weight and nearly equal rates. 
It therefore isolates, to a substantial extent, the effect of replacing the \ac{CSS} structure with the proposed XYZ construction.

Moreover, in Fig.~\ref{fig:IS_XYZ_SOTA}, we report the behavior of $\mathcal{C}_{\mathrm{XYZ}-\text{IS}}$ against several state-of-the-art \ac{CSS} \ac{QLDPC} codes instances, with $n \approx 512$, $w_{\mathrm{r}} \approx 8$, and $R \approx 0.02$.
We remark that the performance of $\mathcal{C}_{\mathrm{XYZ}-\text{IS}}$ is better than $\mathcal{C}_{\text{Bic,}2}$, $\mathcal{C}_{\text{HP,}1}$, and $\mathcal{C}_{\text{HP,}2}$, which is consistent with their smaller minimum distances, namely, $d_{\text{Bic,}2} = d_{\text{HP,}1} = 6$, $d_{\text{HP,}2} = 8$. 
We observe that, for $p > 0.09$,  $\mathcal{C}_{\text{GB,}2}$ exhibits poorer performance than the proposed XYZ code.
Instead, we note that, for $p < 0.09$, the performance curve of $\mathcal{C}_{\text{QM}}$ exhibits an error floor.
Eventually, for $p>0.05$, the XYZ code shows better performance than the latter code. 
Moreover, we also include the \ac{LER} performance of $\mathcal{C}_{\text{BB}}$, which is the instance with the biggest $n$ and (estimated) $d$ among the family of \ac{BB} codes~\cite{Bravyi2024_BB_codes}. 
In particular, for \(p>0.07\), the \ac{BB} code achieves a lower \ac{LER} than the XYZ code. Below this crossover point, however, the XYZ code performs better.
Finally, Fig.~\ref{fig:IS_XYZ_SOTA} includes the performance of $\mathcal{C}_{\text{IS}}$, namely the \ac{CSS}  \ac{IS} code constructed using only the tuples $\mathcal{S}_{\mathrm{X}}$ and $\mathcal{S}_{\mathrm{Z}}$ of Section~\ref{subsec:IS_XYZ}. 
This comparison is not rate matched, since $\mathcal{C}_{\text{IS}}$ has a substantially larger dimension than $\mathcal{C}_{\mathrm{XYZ}-\text{IS}}$. 
Rather, it illustrates the effect of adding the $\sym{Y}$-type constraints within the same underlying \ac{IS} construction. Over the simulated range, the resulting XYZ code achieves a lower \ac{LER}, at the cost of a reduced coding rate.

\begin{table}[tb!]
\centering
\caption{Parameters of tested Stabilizer Codes.
With “$\cdot$” in the “\textbf{Ref.}” column we indicate that the corresponding code instance has been 
presented in this work for the first time. 
}
\begin{tabular}{c c c c c}
\hline
\textbf{Code} & \multicolumn{3}{c}{\textbf{Parameters}} & {\textbf{Ref.}} \\
\cline{2-4}
 & $w_{\mathrm{r}}$ & $g_{\text{X}} / g_{\text{Y}} / g_{\text{Z}}$ & $R$ &  \\
\hline
$\mathcal{C}_{\mathrm{XYZ}-\text{QD}}\llbracket 257, 116, 16 \rrbracket$ & $17$ & $4 / 4 / 4$ & $0.45$ & $\cdot$ \\ 
\hline
$\mathcal{C}_{\mathrm{XYZ}-\text{IS}}\llbracket 512, 9, \leq 16 \rrbracket$ & $8$ & $8 / 4 / 8$ & $0.02$ & $\cdot$ \\ 
\hline
$\mathcal{C}_{\text{QD}}\llbracket 257, 121, 10 \rrbracket$ & $17$ & $4 / \cdot / 4$ & $0.47$ & \cite{baldelli2026ISIT} \\
\hline
$\mathcal{C}_{\text{IS}}\llbracket 512, 174, 8 \rrbracket$ & $8$ & $8 / \cdot / 8$ & $0.33$ & \cite{Ostrev2024classicalproduct} \\
\hline
$\mathcal{C}_{\text{Bic,}1}\llbracket 256, 116, 5 \rrbracket$ & $16$ & $4 / \cdot / 4$ & $0.45$ &  $\cdot$ \\ 
\hline
$\mathcal{C}_{\text{Bic,}2}\llbracket 512, 10, 6 \rrbracket$ & $8$ & $4 / \cdot / 4$ & $0.02$ &  $\cdot$ \\ 
\hline
$\mathcal{C}_{\text{GB,}1}\llbracket 272, 142, 8 \rrbracket$ & $16$ & $6 / \cdot / 6$ & $0.52$ &  $\cdot$ \\ 
\hline
$\mathcal{C}_{\text{GB,}2}\llbracket 488, 6,\geq 8 \rrbracket$ & $8$ & $6 / \cdot / 6$ & $0.01$ &  $\cdot$ \\ 
\hline
$\mathcal{C}_{\text{HP,}1}\llbracket 400, 16, 6 \rrbracket$ & $7$ & $6 / \cdot / 6$ & $0.04$ & \cite{Roffe_decoding} \\
\hline
$\mathcal{C}_{\text{HP,}2}\llbracket 625, 25, 8 \rrbracket$ & $7$ & $6 / \cdot / 6$ & $0.04$ & \cite{Roffe_decoding} \\
\hline
$\mathcal{C}_{\text{QM}}\llbracket 672, 4, \geq 8 \rrbracket$ & $8$ & $6 / \cdot / 6$ & $0.01$ & \cite{Pacenti_Margulis} \\
\hline
$\mathcal{C}_{\text{BB}} \llbracket 756, 16, \leq 34 \rrbracket$ & $6$ & $6 / \cdot / 6$ & $0.02$ & \cite{Bravyi2024_BB_codes} \\
\hline
\end{tabular}
\label{tab:code_parameters}
\end{table}

\section{Conclusions}
\label{sec:conclusions}

We introduced quantum XYZ stabilizer codes as a generalization of the \ac{CSS} framework obtained by including $Y$-type stabilizer generators. 
We characterized when these additional generators produce a genuinely non-\ac{CSS} code and derived bounds on the resulting minimum distance. 
Moreover, we show that the topological XYZ$^2$ code fits in the proposed XYZ framework.
The constructions based on \ac{IS} and \ac{QD} codes show that this framework can yield sparse finite-length codes, while the numerical results confirm that the proposed instances can outperform comparable \ac{CSS} codes under \ac{BP4} decoding.

Future work should focus on the joint design of the three component codes and on decoders tailored to the full XYZ structure. 
It would also be relevant to extend the analysis beyond the code-capacity setting and to clarify the complexity of recognizing \ac{LC} equivalence to \ac{CSS} codes.

\section*{Acknowledgments}

Davide Orsucci acknowledges support in preliminary technical and conceptual work from Dr.\ Francesco Guatieri.

\bibliographystyle{IEEEtran}
\bibliography{strings, Archive_short}

\end{document}